\documentclass[a4paper,USenglish, thm-restate]{lipics-v2019}

\hideLIPIcs
\nolinenumbers

\usepackage{needspace}
\usepackage{comment} 

\usepackage{libertine} 
\usepackage{inconsolata} 

\usepackage{amsmath, amsthm, amssymb}
\usepackage{stmaryrd}
\usepackage{graphicx}
\usepackage{enumerate}

\usepackage[textsize=small, textwidth=2cm, color=yellow]{todonotes}
\usepackage{thmtools}

\usepackage{float}
\usepackage{xcolor}
\usepackage{doi} 
\usepackage[numbers, sort&compress]{natbib} 

\declaretheorem[name=Observation, style=remark, sibling=theorem]{observation}

\def\cqedsymbol{\ifmmode$\lrcorner$\else{\unskip\nobreak\hfil
\penalty50\hskip1em\null\nobreak\hfil$\lrcorner$
\parfillskip=0pt\finalhyphendemerits=0\endgraf}\fi}

\newcommand{\logsn}{\log^*\!n}

\DeclareMathOperator{\poly}{poly}

\let\le\leqslant
\let\ge\geqslant
\let\leq\leqslant
\let\geq\geqslant

\title{Distributed recoloring of interval and chordal graphs}  

\author{Nicolas Bousquet}{Univ. Lyon, Université Lyon 1, LIRIS UMR CNRS 5205, F-69621, Lyon, France}{nicolas.bousquet@univ-lyon1.fr}{https://orcid.org/0000-0003-0170-0503}{}

\author{Laurent Feuilloley}{Univ. Lyon, Université Lyon 1, LIRIS UMR CNRS 5205, F-69621, Lyon, France}{laurent.feuilloley@univ-lyon1.fr}{https://orcid.org/0000-0002-3994-0898}{}

\author{Marc Heinrich}{University of Leeds, UK}{M.Heinrich@leeds.ac.uk }{0000-0003-4546-2359}{}

\author{Mikaël Rabie}
{Université de Paris, CNRS, IRIF, F-75006, Paris, France, France}{mikael.rabie@irif.fr}{https://orcid.org/0000-0001-6782-7625}{}

\authorrunning{N. Bousquet, L. Feuilloley, M. Heinrich and M. Rabie}

\funding{This work was supported by ANR project GrR (ANR-18-CE40-0032).}

\ccsdesc[100]{Theory of computation $\rightarrow$ Design and analysis of algorithms $\rightarrow$ Distributed algorithms}

\keywords{Distributed coloring, distributed recoloring, interval graphs, chordal graphs, intersection graphs} 

\begin{document}

\maketitle

\begin{abstract}

One of the fundamental and most-studied algorithmic problems in distributed computing on networks is graph coloring, both in bounded-degree and in general graphs. 
Recently, the study of this problem has been extended in two directions. 
First, the problem of recoloring, that is computing an efficient transformation between two given colorings (instead of computing a new coloring), has been considered, both to model radio network updates, and as a useful subroutine for coloring. 
Second, as it appears that general graphs and bounded-degree graphs do not model real networks very well (with, respectively, pathological worst-case topologies and too strong assumptions), coloring has been studied in more specific graph classes. 
In this paper, we study the intersection of these two directions: distributed recoloring in two relevant graph classes, interval and chordal graphs.

More formally, the question of recoloring a graph is as follows: we are given a network, an input coloring $\alpha$ and a target coloring $\beta$, and we want to find a schedule of colorings to reach $\beta$ starting from $\alpha$. 
In a distributed setting, the schedule needs to be found within the LOCAL model, where nodes communicate with their direct neighbors synchronously. 
The question we want to answer is: how many rounds of communication {are} needed to produce a schedule, and what is the length of this schedule?

In the case of interval and chordal graphs, we prove that, if we have less than $2\omega$ colors, $\omega$ being the size of the largest clique, extra colors will be needed in the intermediate colorings. For interval graphs, we produce a schedule after $O(\poly(\Delta)\logsn)$ rounds of communication, and for chordal graphs, we need $O(\omega^2\Delta^2\log n)$ rounds to get one.

Our techniques also improve classic coloring algorithms. 
Namely, we get $\omega+1$-colorings of interval graphs in $O(\omega\log^*\!n)$ rounds and of chordal graphs in $O(\omega\log n)$ rounds, which improves on previous known algorithms that use $\omega+2$ colors for the same running times.

\end{abstract}


\newpage{}

\section{Introduction}

Finding a proper coloring of the network is one of the most studied problems in the LOCAL model of distributed computing. 
It is a typical symmetry-breaking problem, and it can be used as a building block to solve many other problems \cite{BarenboimE13}. 
It also has direct applications, and a popular one since~\cite{Ramanathan99} is the assignment of frequencies (or time slots) in wireless networks. 
Consider a network of stations using radio communication. A simple model consists in considering that either two stations are close enough to communicate, but then they should use different emission frequencies to avoid
collisions, or they are far apart, cannot communicate, and can safely use the same frequency.
An assignment of frequencies that avoids conflict is equivalent to a proper coloring. 

Simply computing a coloring is not enough in some contexts. Indeed, suppose that we have an algorithm for the classic coloring task: starting from a coloring with a high number of colors (for example, unique identifiers), we get a coloring with fewer colors. Now, how do we update the frequencies of the network? 
If we do it simultaneously without stopping all communications, there might be conflicts during the transition period, between the nodes still using the old frequencies and the ones using the new frequencies.
Thus we need to stop all communications during the transition period, which is inconvenient if the network is performing other tasks in parallel.  
Therefore, instead of just a new coloring, we would like to have a series of colorings, such that at any step the color change is safe, that is old and new colors do not conflict. 
In this paper, we aim at finding such schedules in a distributed way. More precisely, we tackle the following more general recoloring question: how to go from an input coloring to a target coloring, such that at each step, the coloring is proper and the vertices whose colors are changing form an independent set of the network. To do so, it will sometimes be necessary to assume that we are allowed to use a few additional colors, that are not present in the input and target colorings.

For this problem, three questions naturally arise: Do we need extra colors to be able to produce a schedule, and if yes, how many? Can we find a new schedule locally? And how long the recoloring schedule needs to be? 
We study the problem from the point of view of the LOCAL model, thus we want our algorithms to use a sublinear number of communication rounds.
For this setting, it is known that in general graphs we sometimes need to use many additional colors (\emph{e.g.} $k-1$ extra colors to go from a $k$-coloring to another).
This is bad news, as using extra colors can be considered as costly (\emph{e.g.} because we need to reserve frequencies to make the transitions).
Fortunately, radio networks have topologies that are more constrained than general graphs, and we want to exploit these properties to design schedules that use less extra colors. 

In a first approximation, radio networks can be modeled as intersection graphs of (unit) disks in the plane (see Figure~\ref{fig:intersection-graphs}). In a unit disk graphs, we associate to each node a location on the plane, and two nodes are in conflict (i.e., linked by an edge) if their (unit) disks intersect.  
Unit disks can always be colored with $3\omega$ colors~\cite{peeters1991coloring}, where $\omega$ is the size of the maximum clique of the graph. In contrast, there exists triangle-free graphs with arbitrarily large chromatic number. However, deciding if there exists a $k$-coloring remains NP-complete in the centralized setting, even in unit disks graphs. In general, we do not have a good understanding of the coloring of unit-disk graphs, and we know very little about recoloring, even in the centralized setting.

One can then wonder what happens in other simpler geometric graph classes. For instance, if we assume that all the centers are on the same line in the plane, the obtained graphs are called \emph{interval graphs}. An interval graph is an intersection graph of intervals of the real line (see Figure~\ref{fig:intersection-graphs}). In other words, an interval of the real line is associated to every node of the graph and there is an edge between two nodes if their corresponding intervals intersect.

Interval graphs can be colored with $\omega$ colors in polynomial time in the centralized setting.

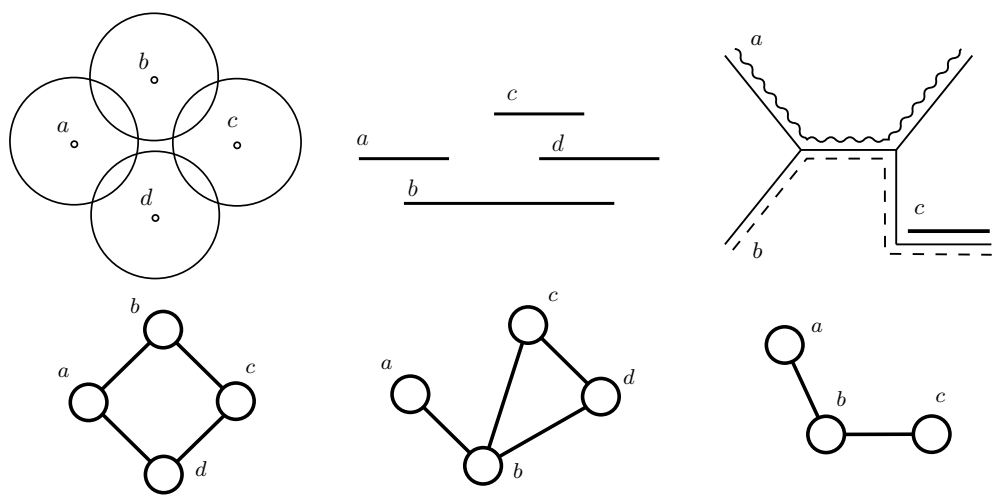
\begin{figure}[h!]
    \centering
    \begin{tabular}{ccc}
    \begin{minipage}{0.3\textwidth}
    \centering
        \scalebox{0.85}{
     \tikzset{every picture/.style={line width=0.75pt}} 

\begin{tikzpicture}[x=0.75pt,y=0.75pt,yscale=-1,xscale=1]

\draw   (172,149.65) .. controls (172,128.86) and (188.86,112) .. (209.65,112) .. controls (230.44,112) and (247.3,128.86) .. (247.3,149.65) .. controls (247.3,170.44) and (230.44,187.3) .. (209.65,187.3) .. controls (188.86,187.3) and (172,170.44) .. (172,149.65) -- cycle ;
\draw   (219,111.65) .. controls (219,90.86) and (235.86,74) .. (256.65,74) .. controls (277.44,74) and (294.3,90.86) .. (294.3,111.65) .. controls (294.3,132.44) and (277.44,149.3) .. (256.65,149.3) .. controls (235.86,149.3) and (219,132.44) .. (219,111.65) -- cycle ;
\draw   (267.65,150.3) .. controls (267.65,129.51) and (284.51,112.65) .. (305.3,112.65) .. controls (326.09,112.65) and (342.95,129.51) .. (342.95,150.3) .. controls (342.95,171.09) and (326.09,187.95) .. (305.3,187.95) .. controls (284.51,187.95) and (267.65,171.09) .. (267.65,150.3) -- cycle ;
\draw   (219.65,193.3) .. controls (219.65,172.51) and (236.51,155.65) .. (257.3,155.65) .. controls (278.09,155.65) and (294.95,172.51) .. (294.95,193.3) .. controls (294.95,214.09) and (278.09,230.95) .. (257.3,230.95) .. controls (236.51,230.95) and (219.65,214.09) .. (219.65,193.3) -- cycle ;
\draw   (207.98,151.62) .. controls (207.9,150.62) and (208.65,149.73) .. (209.65,149.65) .. controls (210.65,149.57) and (211.54,150.31) .. (211.62,151.32) .. controls (211.7,152.32) and (210.96,153.2) .. (209.95,153.29) .. controls (208.95,153.37) and (208.07,152.62) .. (207.98,151.62) -- cycle ;
\draw   (254.98,113.62) .. controls (254.9,112.62) and (255.65,111.73) .. (256.65,111.65) .. controls (257.65,111.57) and (258.54,112.31) .. (258.62,113.32) .. controls (258.7,114.32) and (257.96,115.2) .. (256.95,115.29) .. controls (255.95,115.37) and (255.07,114.62) .. (254.98,113.62) -- cycle ;
\draw   (303.63,152.27) .. controls (303.55,151.27) and (304.3,150.38) .. (305.3,150.3) .. controls (306.3,150.22) and (307.19,150.96) .. (307.27,151.97) .. controls (307.35,152.97) and (306.61,153.85) .. (305.6,153.94) .. controls (304.6,154.02) and (303.72,153.27) .. (303.63,152.27) -- cycle ;
\draw   (255.63,195.27) .. controls (255.55,194.27) and (256.3,193.38) .. (257.3,193.3) .. controls (258.3,193.22) and (259.19,193.96) .. (259.27,194.97) .. controls (259.35,195.97) and (258.61,196.85) .. (257.6,196.94) .. controls (256.6,197.02) and (255.72,196.27) .. (255.63,195.27) -- cycle ;

\draw (198,137) node [anchor=north west][inner sep=0.75pt]   [align=left] {{\large $a$}};
\draw (246,96) node [anchor=north west][inner sep=0.75pt]   [align=left] {{\large $b$}};
\draw (298,134) node [anchor=north west][inner sep=0.75pt]   [align=left] {{\large $c$}};
\draw (247,176) node [anchor=north west][inner sep=0.75pt]   [align=left] {{\large $d$}};

\end{tikzpicture}}
    \end{minipage}
         &  
    \begin{minipage}{0.3\textwidth}
    \centering
        \scalebox{0.75}{
     \tikzset{every picture/.style={line width=0.75pt}} 

\begin{tikzpicture}[x=0.75pt,y=0.75pt,yscale=-1,xscale=1]

\draw [line width=1.5]    (140,80) -- (200,80) ;
\draw [line width=1.5]    (170,110) -- (310,110) ;
\draw [line width=1.5]    (230,50) -- (290,50) ;
\draw [line width=1.5]    (260,80) -- (340,80) ;

\draw (137,62.4) node [anchor=north west][inner sep=0.75pt]    {{\Large $a$}};
\draw (171,92.4) node [anchor=north west][inner sep=0.75pt]    {{\Large $b$}};
\draw (237,32.4) node [anchor=north west][inner sep=0.75pt]    {{\Large $c$}};
\draw (267,62.4) node [anchor=north west][inner sep=0.75pt]    {{\Large $d$}};

\end{tikzpicture}}
    \end{minipage}     
         & 
         \begin{minipage}{0.3\textwidth}
         \centering
        \scalebox{0.95}{
     \tikzset{every picture/.style={line width=0.75pt}} 

\begin{tikzpicture}[x=0.75pt,y=0.75pt,yscale=-1,xscale=1]

\draw    (110,50) -- (150,100) ;
\draw    (150,100) -- (110,150) ;
\draw    (200,100) -- (150,100) ;
\draw    (200,150) -- (200,100) ;
\draw    (200,100) -- (240,50) ;
\draw    (250,150) -- (200,150) ;
\draw    (115.17,46.5) .. controls (117.51,46.77) and (118.54,48.08) .. (118.27,50.42) .. controls (118,52.76) and (119.03,54.07) .. (121.37,54.34) .. controls (123.71,54.61) and (124.75,55.92) .. (124.48,58.26) .. controls (124.21,60.6) and (125.24,61.91) .. (127.58,62.18) .. controls (129.92,62.45) and (130.95,63.76) .. (130.68,66.1) .. controls (130.41,68.44) and (131.45,69.75) .. (133.79,70.02) .. controls (136.13,70.29) and (137.16,71.6) .. (136.89,73.94) .. controls (136.62,76.28) and (137.65,77.59) .. (139.99,77.86) .. controls (142.33,78.13) and (143.37,79.44) .. (143.1,81.78) .. controls (142.83,84.12) and (143.86,85.43) .. (146.2,85.7) .. controls (148.54,85.97) and (149.58,87.28) .. (149.31,89.62) .. controls (149.04,91.96) and (150.07,93.27) .. (152.41,93.54) -- (153.17,94.5) -- (153.17,94.5) .. controls (154.84,92.83) and (156.5,92.83) .. (158.17,94.5) .. controls (159.84,96.17) and (161.5,96.17) .. (163.17,94.5) .. controls (164.84,92.83) and (166.5,92.83) .. (168.17,94.5) .. controls (169.84,96.17) and (171.5,96.17) .. (173.17,94.5) .. controls (174.84,92.83) and (176.5,92.83) .. (178.17,94.5) .. controls (179.84,96.17) and (181.5,96.17) .. (183.17,94.5) .. controls (184.84,92.83) and (186.5,92.83) .. (188.17,94.5) .. controls (189.84,96.17) and (191.5,96.17) .. (193.17,94.5) -- (196.17,94.5) -- (196.17,94.5) .. controls (195.9,92.16) and (196.93,90.85) .. (199.27,90.58) .. controls (201.61,90.31) and (202.64,89) .. (202.37,86.66) .. controls (202.1,84.32) and (203.14,83.01) .. (205.48,82.74) .. controls (207.82,82.47) and (208.85,81.16) .. (208.58,78.82) .. controls (208.31,76.48) and (209.34,75.17) .. (211.68,74.9) .. controls (214.02,74.63) and (215.06,73.32) .. (214.79,70.98) .. controls (214.52,68.64) and (215.55,67.33) .. (217.89,67.06) .. controls (220.23,66.79) and (221.26,65.48) .. (220.99,63.14) .. controls (220.72,60.8) and (221.76,59.49) .. (224.1,59.22) .. controls (226.44,58.95) and (227.47,57.64) .. (227.2,55.3) .. controls (226.93,52.96) and (227.97,51.65) .. (230.31,51.38) .. controls (232.65,51.11) and (233.68,49.8) .. (233.41,47.46) -- (234.17,46.5) -- (234.17,46.5) ;
\draw  [dash pattern={on 4.5pt off 4.5pt}]  (114.13,152.98) -- (153,104.58) -- (194,104.58) -- (194.13,154.98) -- (250.17,155.5) ;
\draw [line width=1.5]    (206.13,142.98) -- (222.13,142.98) -- (249.13,142.98) ;

\draw (122,38) node [anchor=north west][inner sep=0.75pt]    {$a$};
\draw (123,147.4) node [anchor=north west][inner sep=0.75pt]    {$b$};
\draw (208,127) node [anchor=north west][inner sep=0.75pt]    {$c$};

\end{tikzpicture}}
    \end{minipage}
    \\
    \begin{minipage}{0.3\textwidth}
        \centering
        \vspace{0.2cm}
        \scalebox{0.9}{
     \tikzset{every picture/.style={line width=0.75pt}} 

\begin{tikzpicture}[x=0.75pt,y=0.75pt,yscale=-1,xscale=1]

\draw [line width=1.5]    (108.48,150.82) -- (149.82,110.15) ;
\draw [line width=1.5]    (150.15,191.15) -- (191.48,150.48) ;
\draw [line width=1.5]    (149.82,110.15) -- (190.15,150.15) ;
\draw [line width=1.5]    (109.82,151.15) -- (150.15,191.15) ;
\draw  [fill={rgb, 255:red, 255; green, 255; blue, 255 }  ,fill opacity=1 ][line width=1.5]  (180,150.15) .. controls (180,144.54) and (184.54,140) .. (190.15,140) .. controls (195.76,140) and (200.3,144.54) .. (200.3,150.15) .. controls (200.3,155.76) and (195.76,160.3) .. (190.15,160.3) .. controls (184.54,160.3) and (180,155.76) .. (180,150.15) -- cycle ;
\draw  [fill={rgb, 255:red, 255; green, 255; blue, 255 }  ,fill opacity=1 ][line width=1.5]  (140,191.15) .. controls (140,185.54) and (144.54,181) .. (150.15,181) .. controls (155.76,181) and (160.3,185.54) .. (160.3,191.15) .. controls (160.3,196.76) and (155.76,201.3) .. (150.15,201.3) .. controls (144.54,201.3) and (140,196.76) .. (140,191.15) -- cycle ;
\draw  [fill={rgb, 255:red, 255; green, 255; blue, 255 }  ,fill opacity=1 ][line width=1.5]  (98.33,150.82) .. controls (98.33,145.21) and (102.88,140.67) .. (108.48,140.67) .. controls (114.09,140.67) and (118.63,145.21) .. (118.63,150.82) .. controls (118.63,156.42) and (114.09,160.97) .. (108.48,160.97) .. controls (102.88,160.97) and (98.33,156.42) .. (98.33,150.82) -- cycle ;
\draw  [fill={rgb, 255:red, 255; green, 255; blue, 255 }  ,fill opacity=1 ][line width=1.5]  (139.67,110.15) .. controls (139.67,104.54) and (144.21,100) .. (149.82,100) .. controls (155.42,100) and (159.97,104.54) .. (159.97,110.15) .. controls (159.97,115.76) and (155.42,120.3) .. (149.82,120.3) .. controls (144.21,120.3) and (139.67,115.76) .. (139.67,110.15) -- cycle ;

\draw (90,130) node [anchor=north west][inner sep=0.75pt]    {$a$};
\draw (130,91) node [anchor=north west][inner sep=0.75pt]    {$b$};
\draw (194,128) node [anchor=north west][inner sep=0.75pt]    {$c$};
\draw (166,182.4) node [anchor=north west][inner sep=0.75pt]    {$d$};

\end{tikzpicture}}
    \end{minipage}
         &  
    \begin{minipage}{0.3\textwidth}
    \centering
        \scalebox{0.9}{
     \tikzset{every picture/.style={line width=0.75pt}} 

\begin{tikzpicture}[x=0.75pt,y=0.75pt,yscale=-1,xscale=1]

\draw [line width=1.5]    (124.82,188.82) -- (149.82,110.15) ;
\draw [line width=1.5]    (124.82,188.82) -- (191.48,150.48) ;
\draw [line width=1.5]    (149.82,110.15) -- (190.15,150.15) ;
\draw [line width=1.5]    (84.48,148.82) -- (124.82,188.82) ;
\draw  [fill={rgb, 255:red, 255; green, 255; blue, 255 }  ,fill opacity=1 ][line width=1.5]  (180,150.15) .. controls (180,144.54) and (184.54,140) .. (190.15,140) .. controls (195.76,140) and (200.3,144.54) .. (200.3,150.15) .. controls (200.3,155.76) and (195.76,160.3) .. (190.15,160.3) .. controls (184.54,160.3) and (180,155.76) .. (180,150.15) -- cycle ;
\draw  [fill={rgb, 255:red, 255; green, 255; blue, 255 }  ,fill opacity=1 ][line width=1.5]  (114.67,188.82) .. controls (114.67,183.21) and (119.21,178.67) .. (124.82,178.67) .. controls (130.42,178.67) and (134.97,183.21) .. (134.97,188.82) .. controls (134.97,194.42) and (130.42,198.97) .. (124.82,198.97) .. controls (119.21,198.97) and (114.67,194.42) .. (114.67,188.82) -- cycle ;
\draw  [fill={rgb, 255:red, 255; green, 255; blue, 255 }  ,fill opacity=1 ][line width=1.5]  (74.33,148.82) .. controls (74.33,143.21) and (78.88,138.67) .. (84.48,138.67) .. controls (90.09,138.67) and (94.63,143.21) .. (94.63,148.82) .. controls (94.63,154.42) and (90.09,158.97) .. (84.48,158.97) .. controls (78.88,158.97) and (74.33,154.42) .. (74.33,148.82) -- cycle ;
\draw  [fill={rgb, 255:red, 255; green, 255; blue, 255 }  ,fill opacity=1 ][line width=1.5]  (139.67,110.15) .. controls (139.67,104.54) and (144.21,100) .. (149.82,100) .. controls (155.42,100) and (159.97,104.54) .. (159.97,110.15) .. controls (159.97,115.76) and (155.42,120.3) .. (149.82,120.3) .. controls (144.21,120.3) and (139.67,115.76) .. (139.67,110.15) -- cycle ;

\draw (66.67,128) node [anchor=north west][inner sep=0.75pt]    {$a$};
\draw (140,185.73) node [anchor=north west][inner sep=0.75pt]    {$b$};
\draw (159.33,90) node [anchor=north west][inner sep=0.75pt]    {$c$};
\draw (201,133.07) node [anchor=north west][inner sep=0.75pt]    {$d$};

\end{tikzpicture}}
    \end{minipage}     
         & 
    \begin{minipage}{0.3\textwidth}
    \centering
        \scalebox{0.9}{
     \tikzset{every picture/.style={line width=0.75pt}} 

\begin{tikzpicture}[x=0.75pt,y=0.75pt,yscale=-1,xscale=1]

\draw [line width=1.5]    (124.82,188.82) -- (183.33,188.67) ;
\draw [line width=1.5]    (97.82,130.15) -- (124.82,188.82) ;
\draw  [fill={rgb, 255:red, 255; green, 255; blue, 255 }  ,fill opacity=1 ][line width=1.5]  (173.18,188.67) .. controls (173.18,183.06) and (177.73,178.52) .. (183.33,178.52) .. controls (188.94,178.52) and (193.48,183.06) .. (193.48,188.67) .. controls (193.48,194.27) and (188.94,198.82) .. (183.33,198.82) .. controls (177.73,198.82) and (173.18,194.27) .. (173.18,188.67) -- cycle ;
\draw  [fill={rgb, 255:red, 255; green, 255; blue, 255 }  ,fill opacity=1 ][line width=1.5]  (114.67,188.82) .. controls (114.67,183.21) and (119.21,178.67) .. (124.82,178.67) .. controls (130.42,178.67) and (134.97,183.21) .. (134.97,188.82) .. controls (134.97,194.42) and (130.42,198.97) .. (124.82,198.97) .. controls (119.21,198.97) and (114.67,194.42) .. (114.67,188.82) -- cycle ;
\draw  [fill={rgb, 255:red, 255; green, 255; blue, 255 }  ,fill opacity=1 ][line width=1.5]  (91.67,138.82) .. controls (91.67,133.21) and (96.21,128.67) .. (101.82,128.67) .. controls (107.42,128.67) and (111.97,133.21) .. (111.97,138.82) .. controls (111.97,144.42) and (107.42,148.97) .. (101.82,148.97) .. controls (96.21,148.97) and (91.67,144.42) .. (91.67,138.82) -- cycle ;

\draw (115,125.07) node [anchor=north west][inner sep=0.75pt]    {$a$};
\draw (129,163) node [anchor=north west][inner sep=0.75pt]    {$b$};
\draw (184,164) node [anchor=north west][inner sep=0.75pt]    {$c$};

\end{tikzpicture}}
    \end{minipage} 
    \end{tabular}
        \caption{Illustrations of the different intersection graph classes mentioned. On each column, the geometric intersection model is on top, and the corresponding graph on the bottom.
        The left-most column is for unit-disk graph, the middle one for interval graphs, and the right-most for chordal graphs. For the chordal intersection model, the base tree uses plain edges, $a$ is the wavy subtree, $b$ is the dashed one, and $c$ the bold one. (This last intersection model is simplistic, to allow a readable drawing. In general the subtrees are not paths, and for this specific graph, a simpler geometric representation exists.)}
    \label{fig:intersection-graphs}
\end{figure}

A natural generalization of interval graphs are chordal graphs. Chordal graphs are one of the most-well studied classes in graph theory, with deep structures, fast algorithms, and many important applications \cite{Golumbic80, BlairP93, VandenbergheA15}. 
They are intersection graphs of subtrees of a tree. In other words, given a tree $T$, every vertex of a chordal graph $G$ is associated to a subtree of $T$ and two nodes of $G$ are adjacent if their corresponding subtrees intersect (see Figure~\ref{fig:intersection-graphs}).

Our main results are recoloring algorithms for interval and chordal graphs. Our techniques also yield coloring algorithms that improve on the state of the art. 

\subsection{Our results}

Our main results are recoloring algorithms for interval and chordal graphs. 
Let us describe a bit more formally what a distributed recoloring algorithm is. We use the definition introduced in~\cite{bonamy2018distributed}. A \emph{valid recoloring schedule} from a coloring $\alpha$ to a coloring $\beta$ consists in a sequence of colorings that starts with $\alpha$ and ends in $\beta$ such that, at every step, the coloring is proper and the vertices whose colors are modified at a given step (called the \emph{recolored vertices}) form an independent set of the graph. A distributed recoloring algorithm is an algorithm such that every node computes its own schedule. In this paper, we compute the schedule in the LOCAL model. Each node can check the validity of the schedule by comparing its own with its neighbors: they check that at each step, the coloring is locally proper, and that if it changes its own color during a step, none of its neighbors does the same.

In this paper, we will study how many \emph{rounds} are used in the LOCAL model to produce a schedule of some \emph{length}.
We first focus on interval graphs and then extend our result to chordal graphs.
We first prove the following:

\begin{restatable}{theorem}{ThmIntervalle}
\label{thm:recol-interval}
Let $G$ be an interval graph and $\alpha,\beta$ be two proper $k$-colorings of $G$. It is possible to find a schedule to transform $\alpha$ into $\beta$  in the LOCAL model in $O(\poly(\Delta) \logsn)$ rounds using at most:
\begin{itemize}
    \item $c$ additional colors, with $c = \omega-k+4$, if $k \le \omega+2$, with a schedule of length $\poly(\Delta)$,
    \item $1$ additional color if $k \ge \omega+3$, with a schedule of length $\poly(\Delta)$,
    \item no additional color if $k \ge 2\omega$ with a schedule of exponential-in-$\Delta$ length.
    \item no additional color if $k \ge 4\omega$ with a schedule of length $O(\omega\Delta)$.
\end{itemize}
\end{restatable}

\medskip

We also prove in Lemma~\ref{lem:lower-bound} that the complexity of second item of Theorem~\ref{thm:recol-interval} is optimal, in the sense that with less than $2 \omega$ colors,
the number of rounds (as well as a schedule) must be linear in the size of the graph.
Basically, when the number of colors is smaller than $2\omega$, if we do not have additional colors, we might need to recolor vertices from the border of a graph in order to be able to recolor vertices in the middle. On the contrary, when the number is at least $2\omega$, we can save a color using only local modifications, and then proceed as if we had one additional unused color. 
Unfortunately, in order to free this additional color (Theorem~\ref{Thm2Omega}) we use a schedule of size exponential  in $\Delta$.\footnote{It is likely that the $O(n^{d+1})$-recoloring algorithm in the centralized setting for $d$-degenerate graphs can be adapted in order to provide a polynomial schedule instead, but we did not do it to keep the proof as short as possible.} 
On the positive side, this color saving step (and actually the whole recoloring) can be performed with a short schedule when $k \ge 4\omega$.

For the first item, as argued above, at least one additional is needed.
The colors that we use in addition to that one come from a result of~\cite{BousquetB19} that we use (almost) as a black-box (see the next section). 
If the number of colors could be decreased in the context of~\cite{BousquetB19}, then it could also be decreased in our work.
We left as an open problem the question of deciding if the numbers of additional colors can be reduced to 1. 

One can naturally wonder what are the dependencies in $\Delta$ and $\omega$ on our complexity and schedule lengths, as we did not give them explicitly. 
We are using as a black-box the recoloring result on chordal graphs of~\cite{BousquetB19}, which gives a centralized recoloring algorithm in $O(\omega^4 \cdot \Delta \cdot n)$ steps, where only one vertex can be recolored at each step. 
The schedule we produce contains as subsequence a constant number of such black-box schedules, corresponding to subgraphs  
containing about $O(\omega^4 \cdot \Delta^2)$ vertices. 
In order to keep the proofs as simple as possible, we did not try to optimize this polynomial function. 
In particular, the recoloring procedure of~\cite{BousquetB19} can probably be adapted in the LOCAL model with a parallel schedule shorter than $O(\omega^4 \cdot \Delta \cdot n)$.

Since for interval graphs, the gap between $\omega$ and $\Delta$ can be arbitrarily large, it would be interesting to determine if the recoloring schedule (as well as the number of rounds) can be reduced to a function of $\omega$ only.

\medskip

Our second main result consists in extending this recoloring result to chordal graphs. Namely, we prove:

\begin{restatable}{theorem}{ThmChordal}
\label{thm:recol-chordal}
Let $G$ be a chordal graph and $\alpha,\beta$ be two proper $k$-colorings of $G$. It is possible to find a schedule of length $n^{O(\log \Delta)}$ to transform $\alpha$ into $\beta$ in $O(\omega^2 \Delta^2 \log n)$ rounds in the LOCAL model using at most:
\begin{itemize}
    \item $c$ additional colors, with $c = \omega-k+4$, if $k \le \omega+2$,
    \item $1$ additional color if $k \ge \omega+3$.
\end{itemize}
\end{restatable}

\medskip

Note that while the recoloring schedule of Theorem~\ref{thm:recol-interval} is polynomial, the one obtained in Theorem~\ref{thm:recol-chordal} is superpolynomial in $\Delta$ (but the number of rounds remains polynomial).  We left as an open problem the existence of polynomial schedule.
In terms of number of rounds, we also pay a logarithmic factor in comparison to the $O(\logsn)$ rounds used for interval graphs. It would be interesting to know if there exists a $(\omega+x)$-coloring algorithm for chordal graphs for a constant $x$, that only needs $O(\log^* n)$ rounds.
We will explain in the subsection that follows, that these differences in schedule length and running time are inherent to our approach, and going beyond those would require new techniques.

\medskip

While designing tools for recoloring, we also produce improved algorithms for the coloring problem. In a nutshell, it was known how to get $(\omega+2)$-colorings for interval and chordal graphs, and we improve this to $(\omega+1)$ colors.
More precisely, it was proved in~\cite{HalldorssonK20} that an $(\omega+2)$-coloring of interval graphs can be obtained in the LOCAL model in $O(\omega \logsn)$ rounds.\footnote{Actually, in \cite{HalldorssonK20} the result is stated as a $(1+\epsilon)$-approximation of the optimal $\omega$-coloring, with the condition that $\epsilon\geq 2/\omega$, which is equivalent to an $(\omega+2)$-coloring.} Our first result consists in improving the result of Halldorsson and Konrad~\cite{HalldorssonK20} by proving:

\begin{restatable}{theorem}{ThmColInterval}
\label{thm:col-interval}
Interval graphs can be colored with $(\omega+1)$-colors in $O(\omega \logsn)$ rounds in the LOCAL model.
\end{restatable}

Note that $(\omega+1)$ is the best we can hope for in sublinear time in paths (which are interval graphs with $\omega=2$), since paths cannot be colored with $2$ colors in less than $\Omega(n)$ rounds~\cite{Linial92, BarenboimE13}. 
This can actually be generalized for any fixed $\omega$, by considering the $(\omega-1)$-th power of a path. The \emph{$k$-th power of a path} is a graph with vertex set $v_1,\ldots,v_n$ where two vertices $v_i,v_j$ are adjacent if and only if $|i-j| \le k$. Hence, each maximal clique corresponds to $\omega$ consecutive nodes. 
An $\omega$-coloring of such a graph would require the algorithm to put a same color every $\omega$ nodes in the path, which again requires $\Omega(n)$ communication rounds.

For chordal graphs, Konrad and Zamaraev~\cite{KonradZ19} proved that an $(\omega+2)$-coloring can be found in $O(\omega \log n)$ steps. Again, we can reduce the number of colors to $\omega+1$. 

\begin{restatable}{theorem}{ThmColChordal}
\label{thm:col-chordal}
Chordal graphs can be colored with $(\omega+1)$-colors in $O(\omega \log n)$ rounds in the LOCAL model.
\end{restatable}

Again, the dependency in $n$ is optimal for $\omega=2$ because (unoriented) trees cannot be 3-colored in $o(\log n)$ rounds \cite{Linial92, BarenboimE13}. We leave as an open question if $\Omega(\log n)$ communication rounds are required to produce an $(\omega+1)$-coloring for $\omega>2$.

\subsection{Our techniques}

Our algorithms are using graph decompositions. 
Basically, we first split the graph into components of controlled diameter, and then work on the different components, taking care of not creating conflicts at the borders.  

To get more into the details, let us start with coloring. For interval graphs, we use a variation of a decomposition from~\cite{HalldorssonK20}. It consists in proving that we can find some subsets of vertices that cut the interval graphs into parts whose diameter is large enough, but not too large, in $O(\logsn)$ rounds (with a multiplicative factor depending on the diameter of each part). See Figure~\ref{fig:interval-decomposition}.
Our contribution is in the second part of the algorithm. 
We show that we can start from an arbitrary coloring of the cuts, and extend this partial coloring into a tight $(\omega+1)$-coloring of $G$, whereas previous algorithms needed some extra slack in the number of colors. 

\begin{figure}[!h]
    \centering
    \scalebox{0.8}{
    \tikzset{every picture/.style={line width=0.75pt}} 

\begin{tikzpicture}[x=0.75pt,y=0.75pt,yscale=-1,xscale=1]

\draw [color={rgb, 255:red, 201; green, 201; blue, 201 }  ,draw opacity=1 ][fill={rgb, 255:red, 155; green, 155; blue, 155 }  ,fill opacity=1 ][line width=4.5]    (490,100) -- (490,220) ;
\draw [color={rgb, 255:red, 201; green, 201; blue, 201 }  ,draw opacity=1 ][fill={rgb, 255:red, 155; green, 155; blue, 155 }  ,fill opacity=1 ][line width=4.5]    (300,100) -- (300,220) ;
\draw [line width=2.25]    (60,150) -- (100,150) ;
\draw [line width=2.25]    (160,130) -- (220,130) ;
\draw [color={rgb, 255:red, 201; green, 201; blue, 201 }  ,draw opacity=1 ][fill={rgb, 255:red, 155; green, 155; blue, 155 }  ,fill opacity=1 ][line width=4.5]    (120,100) -- (120,220) ;
\draw [line width=2.25]    (90,170) -- (180,170) ;
\draw [line width=2.25]    (80,130) -- (130,130) ;
\draw [line width=2.25]    (160,190) -- (220,190) ;
\draw [line width=2.25]    (200,150) -- (237,150) -- (270,150) ;
\draw [line width=2.25]    (250,130) -- (320,130) ;
\draw [line width=2.25]    (250,190) -- (310,190) ;
\draw [line width=2.25]    (280,170) -- (370,170) ;
\draw [line width=2.25]    (405,150) -- (450,150) ;
\draw [line width=2.25]    (350,190) -- (420,190) ;
\draw [line width=2.25]    (340,130) -- (390,130) ;
\draw [line width=2.25]    (430,130) -- (520,130) ;
\draw [line width=2.25]    (470,170) -- (510,170) ;
\draw [line width=2.25]    (470,190) -- (530,190) ;

\end{tikzpicture}}
    \vspace{0.4cm}
    
    \scalebox{0.8}{
    \tikzset{every picture/.style={line width=0.75pt}} 

\begin{tikzpicture}[x=0.75pt,y=0.75pt,yscale=-1,xscale=1]

\draw [line width=1.5]    (60,150) -- (95,130) ;
\draw [line width=1.5]    (60,150) -- (120,170) ;
\draw [line width=1.5]    (95,130) -- (120,170) ;
\draw [line width=1.5]    (180,130) -- (120,170) ;
\draw [line width=1.5]    (180,190) -- (120,170) ;
\draw [line width=1.5]    (218.5,150) -- (180,190) ;
\draw [line width=1.5]    (218.5,150) -- (180,130) ;
\draw [line width=1.5]    (180,130) -- (180,190) ;
\draw [line width=1.5]    (270,130) -- (218.5,150) ;
\draw [line width=1.5]    (270,190) -- (218.5,150) ;
\draw [line width=1.5]    (270,130) -- (270,190) ;
\draw [line width=1.5]    (270,130) -- (310,170) ;
\draw [line width=1.5]    (310,170) -- (270,190) ;
\draw [line width=1.5]    (360,130) -- (310,170) ;
\draw [line width=1.5]    (360,190) -- (310,170) ;
\draw [line width=1.5]    (360,190) -- (395,150) ;
\draw [line width=1.5]    (360,190) -- (360,130) ;
\draw [line width=1.5]    (395,150) -- (450,130) ;
\draw [line width=1.5]    (470,190) -- (450,130) ;
\draw [line width=1.5]    (490,170) -- (450,130) ;
\draw [line width=1.5]    (470,190) -- (490,170) ;
\draw  [fill={rgb, 255:red, 255; green, 255; blue, 255 }  ,fill opacity=1 ][line width=1.5]  (50,150) .. controls (50,144.48) and (54.48,140) .. (60,140) .. controls (65.52,140) and (70,144.48) .. (70,150) .. controls (70,155.52) and (65.52,160) .. (60,160) .. controls (54.48,160) and (50,155.52) .. (50,150) -- cycle ;
\draw  [fill={rgb, 255:red, 155; green, 155; blue, 155 }  ,fill opacity=1 ][line width=1.5]  (85,130) .. controls (85,124.48) and (89.48,120) .. (95,120) .. controls (100.52,120) and (105,124.48) .. (105,130) .. controls (105,135.52) and (100.52,140) .. (95,140) .. controls (89.48,140) and (85,135.52) .. (85,130) -- cycle ;
\draw  [fill={rgb, 255:red, 155; green, 155; blue, 155 }  ,fill opacity=1 ][line width=1.5]  (110,170) .. controls (110,164.48) and (114.48,160) .. (120,160) .. controls (125.52,160) and (130,164.48) .. (130,170) .. controls (130,175.52) and (125.52,180) .. (120,180) .. controls (114.48,180) and (110,175.52) .. (110,170) -- cycle ;
\draw  [fill={rgb, 255:red, 255; green, 255; blue, 255 }  ,fill opacity=1 ][line width=1.5]  (170,130) .. controls (170,124.48) and (174.48,120) .. (180,120) .. controls (185.52,120) and (190,124.48) .. (190,130) .. controls (190,135.52) and (185.52,140) .. (180,140) .. controls (174.48,140) and (170,135.52) .. (170,130) -- cycle ;
\draw  [fill={rgb, 255:red, 255; green, 255; blue, 255 }  ,fill opacity=1 ][line width=1.5]  (170,190) .. controls (170,184.48) and (174.48,180) .. (180,180) .. controls (185.52,180) and (190,184.48) .. (190,190) .. controls (190,195.52) and (185.52,200) .. (180,200) .. controls (174.48,200) and (170,195.52) .. (170,190) -- cycle ;
\draw  [fill={rgb, 255:red, 255; green, 255; blue, 255 }  ,fill opacity=1 ][line width=1.5]  (208.5,150) .. controls (208.5,144.48) and (212.98,140) .. (218.5,140) .. controls (224.02,140) and (228.5,144.48) .. (228.5,150) .. controls (228.5,155.52) and (224.02,160) .. (218.5,160) .. controls (212.98,160) and (208.5,155.52) .. (208.5,150) -- cycle ;
\draw  [fill={rgb, 255:red, 155; green, 155; blue, 155 }  ,fill opacity=1 ][line width=1.5]  (260,130) .. controls (260,124.48) and (264.48,120) .. (270,120) .. controls (275.52,120) and (280,124.48) .. (280,130) .. controls (280,135.52) and (275.52,140) .. (270,140) .. controls (264.48,140) and (260,135.52) .. (260,130) -- cycle ;
\draw  [fill={rgb, 255:red, 155; green, 155; blue, 155 }  ,fill opacity=1 ][line width=1.5]  (260,190) .. controls (260,184.48) and (264.48,180) .. (270,180) .. controls (275.52,180) and (280,184.48) .. (280,190) .. controls (280,195.52) and (275.52,200) .. (270,200) .. controls (264.48,200) and (260,195.52) .. (260,190) -- cycle ;
\draw  [fill={rgb, 255:red, 155; green, 155; blue, 155 }  ,fill opacity=1 ][line width=1.5]  (300,170) .. controls (300,164.48) and (304.48,160) .. (310,160) .. controls (315.52,160) and (320,164.48) .. (320,170) .. controls (320,175.52) and (315.52,180) .. (310,180) .. controls (304.48,180) and (300,175.52) .. (300,170) -- cycle ;
\draw  [fill={rgb, 255:red, 255; green, 255; blue, 255 }  ,fill opacity=1 ][line width=1.5]  (350,130) .. controls (350,124.48) and (354.48,120) .. (360,120) .. controls (365.52,120) and (370,124.48) .. (370,130) .. controls (370,135.52) and (365.52,140) .. (360,140) .. controls (354.48,140) and (350,135.52) .. (350,130) -- cycle ;
\draw  [fill={rgb, 255:red, 255; green, 255; blue, 255 }  ,fill opacity=1 ][line width=1.5]  (350,190) .. controls (350,184.48) and (354.48,180) .. (360,180) .. controls (365.52,180) and (370,184.48) .. (370,190) .. controls (370,195.52) and (365.52,200) .. (360,200) .. controls (354.48,200) and (350,195.52) .. (350,190) -- cycle ;
\draw  [fill={rgb, 255:red, 255; green, 255; blue, 255 }  ,fill opacity=1 ][line width=1.5]  (385,150) .. controls (385,144.48) and (389.48,140) .. (395,140) .. controls (400.52,140) and (405,144.48) .. (405,150) .. controls (405,155.52) and (400.52,160) .. (395,160) .. controls (389.48,160) and (385,155.52) .. (385,150) -- cycle ;
\draw  [fill={rgb, 255:red, 155; green, 155; blue, 155 }  ,fill opacity=1 ][line width=1.5]  (440,130) .. controls (440,124.48) and (444.48,120) .. (450,120) .. controls (455.52,120) and (460,124.48) .. (460,130) .. controls (460,135.52) and (455.52,140) .. (450,140) .. controls (444.48,140) and (440,135.52) .. (440,130) -- cycle ;
\draw  [fill={rgb, 255:red, 155; green, 155; blue, 155 }  ,fill opacity=1 ][line width=1.5]  (480,170) .. controls (480,164.48) and (484.48,160) .. (490,160) .. controls (495.52,160) and (500,164.48) .. (500,170) .. controls (500,175.52) and (495.52,180) .. (490,180) .. controls (484.48,180) and (480,175.52) .. (480,170) -- cycle ;
\draw  [fill={rgb, 255:red, 155; green, 155; blue, 155 }  ,fill opacity=1 ][line width=1.5]  (460,190) .. controls (460,184.48) and (464.48,180) .. (470,180) .. controls (475.52,180) and (480,184.48) .. (480,190) .. controls (480,195.52) and (475.52,200) .. (470,200) .. controls (464.48,200) and (460,195.52) .. (460,190) -- cycle ;

\end{tikzpicture}}
    \caption{Illustration of the interval decomposition of \cite{HalldorssonK20}. On top the interval representation, and on the bottom, the graph. The decomposition consists in choosing cliques (the gray nodes) whose removal decomposes the graphs into subgraphs of controlled diameter (the white nodes). Note that the graph can then be seen as a path alternating between gray cliques and white subgraphs.
    What we do is slightly different, as our cuts are not cliques, but the intuition is the same.}
    \label{fig:interval-decomposition}
\end{figure}
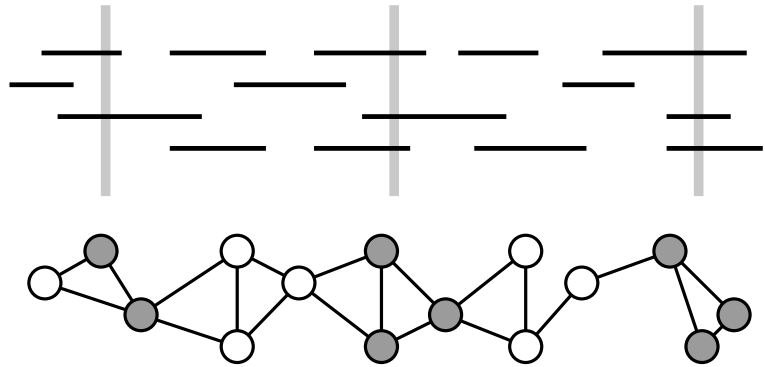

For chordal graph, \cite{KonradZ19} provides a recursive decomposition into interval graphs, that we adapt to our setting. The idea is to use a rake-and-compress strategy (in the spirit of \cite{MillerR89}) to decompose the graph in $O(\log n)$ layers, each layer being a union of interval graphs. 
A typical phase of rake-and-compress on a tree consists in removing the leaves (the rake part), and removing or contracting the long paths (the compress part). 
After $O(\log n)$ such phases, the tree is empty. 
In our algorithm, the tree structure is the underlying tree of the chordal graph, like in Figure~\ref{fig:intersection-graphs}.\footnote{Actually, there are two such trees, the one that correspond to the geometric representation, that we use in this introduction, and the clique tree, which has a more graph-theoretic definition, and that we use in the proofs. The two are essentially equivalent.}
We consider here as leaves the pending interval subgraphs. These are subgraphs made by taking the intervals that are on the branches of the tree leading to a leaf (see Fig.~\ref{fig:chordal-decomposition}). On the other hand, long paths correspond to long separating interval subgraphs, whose removal disconnects the graph, and that have large enough diameter.

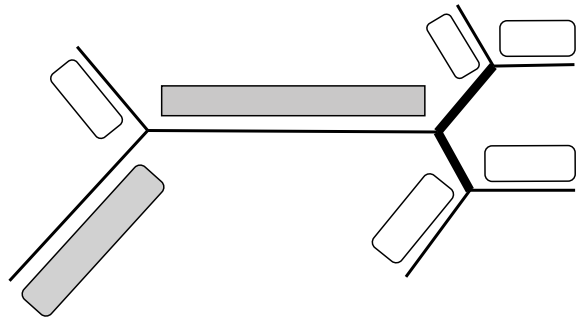
\begin{figure}[!h]
    \centering
    \scalebox{0.75}{
    \tikzset{every picture/.style={line width=0.75pt}} 

\begin{tikzpicture}[x=0.75pt,y=0.75pt,yscale=-1,xscale=1]

\draw [line width=1.5]    (195.57,119.93) -- (388.57,120.93) ;
\draw [line width=1.5]    (148.57,63.73) -- (195.57,119.93) ;
\draw [line width=1.5]    (103.57,220.73) -- (195.57,119.93) ;
\draw [line width=4.5]    (388.57,120.93) -- (425.57,76.73) ;
\draw [line width=1.5]    (425.57,76.73) -- (479.57,75.73) ;
\draw [line width=1.5]    (401.57,35.73) -- (425.57,76.73) ;
\draw [line width=4.5]    (388.57,120.93) -- (410,160) ;
\draw [line width=1.5]    (410,160) -- (398.28,175.99) -- (367.43,218.07) ;
\draw [line width=1.5]    (480,160) -- (410,160) ;
\draw  [fill={rgb, 255:red, 209; green, 209; blue, 209 }  ,fill opacity=1 ] (113.63,233.28) .. controls (111.63,231.47) and (111.47,228.37) .. (113.28,226.37) -- (186.87,144.75) .. controls (188.68,142.74) and (191.77,142.58) .. (193.78,144.39) -- (204.68,154.22) .. controls (206.69,156.03) and (206.85,159.13) .. (205.04,161.14) -- (131.45,242.75) .. controls (129.64,244.76) and (126.55,244.92) .. (124.54,243.11) -- cycle ;
\draw   (346.82,198.56) .. controls (344.77,196.9) and (344.46,193.89) .. (346.12,191.84) -- (379.45,150.75) .. controls (381.11,148.71) and (384.12,148.39) .. (386.17,150.05) -- (397.29,159.07) .. controls (399.33,160.73) and (399.65,163.74) .. (397.99,165.79) -- (364.65,206.88) .. controls (362.99,208.93) and (359.98,209.24) .. (357.94,207.58) -- cycle ;
\draw   (419.98,135.01) .. controls (419.97,132.38) and (422.1,130.23) .. (424.73,130.22) -- (475.23,130.02) .. controls (477.86,130.01) and (480.01,132.14) .. (480.02,134.77) -- (480.08,149.09) .. controls (480.09,151.73) and (477.96,153.87) .. (475.32,153.88) -- (424.83,154.09) .. controls (422.19,154.1) and (420.05,151.97) .. (420.04,149.33) -- cycle ;
\draw   (429.96,51.11) .. controls (429.95,48.47) and (432.08,46.33) .. (434.72,46.32) -- (475.13,46.15) .. controls (477.77,46.14) and (479.91,48.27) .. (479.92,50.91) -- (479.98,65.23) .. controls (479.99,67.86) and (477.86,70.01) .. (475.23,70.02) -- (434.81,70.18) .. controls (432.18,70.19) and (430.03,68.06) .. (430.02,65.43) -- cycle ;
\draw   (392.49,42.47) .. controls (394.21,41.43) and (396.45,41.98) .. (397.49,43.71) -- (415.64,73.78) .. controls (416.68,75.51) and (416.12,77.75) .. (414.4,78.79) -- (405.03,84.44) .. controls (403.31,85.48) and (401.07,84.92) .. (400.03,83.2) -- (381.88,53.12) .. controls (380.84,51.4) and (381.4,49.16) .. (383.12,48.12) -- cycle ;
\draw   (142.61,73.6) .. controls (144.45,72.1) and (147.15,72.39) .. (148.64,74.23) -- (177.72,110.13) .. controls (179.21,111.97) and (178.92,114.68) .. (177.08,116.17) -- (167.08,124.27) .. controls (165.23,125.76) and (162.53,125.48) .. (161.04,123.63) -- (131.97,87.73) .. controls (130.48,85.89) and (130.76,83.19) .. (132.6,81.7) -- cycle ;
\draw  [fill={rgb, 255:red, 201; green, 200; blue, 200 }  ,fill opacity=1 ] (204.84,90) -- (380,90) -- (380,110) -- (204.84,110) -- cycle ;

\end{tikzpicture}}
    \caption{Illustration of one step of the decomposition of chordal graphs into interval graphs (of controlled diameter). The tree of plain edges is the underlying tree. The rectangles with rounded corners correspond to the pending interval graphs. The rectangle with spiky corners correspond to a separating interval graph. The grey rectangles are the ones that are too long and will be further decomposed. The bold edges correspond to part of the graphs that are not pending, nor long and separating. These part are the ones that are kept for the next phases, all the other vertices are removed.}
    \label{fig:chordal-decomposition}
\end{figure}

Let $V_i$ be the set of nodes of the interval graph removed at phase $i$.
The coloring algorithm then consists in coloring recursively $V_k,\ldots,V_1$, in this order. That is, we start by coloring the vertices that have been removed last.
Assume that we have a coloring of $V_{i+1},\ldots,V_k$. 
The key point is that when we consider a new interval subgraph, by construction, only its borders can be already colored, and therefore we are back to the situation we had for interval graphs. We then get the same number of colors. The logarithmic-in-$n$ complexity comes from the $O(\log n)$ layers of the rake-and-compress, that we must process sequentially.
\medskip

Now, let us describe our recoloring techniques.
For interval graphs, Theorem~\ref{thm:recol-interval} is obtained using two tools coming from the centralized setting. First, we use Kempe chains that have been used in several graph coloring proofs in the last decades. Given a graph $G$ and a coloring $c$ of~$G$, an \emph{$(a,b)$-component} is a connected component of $G$ restricted to the vertices colored $a$ and~$b$. 
A \emph{Kempe component} is an $(a,b)$-component for some pair of colors $a,b$. One can remark that, given a proper coloring and an $(a,b)$-component, permuting the colors $a$ and $b$ still leaves a proper coloring. In recoloring, this permutation can be performed by using one extra color as buffer. 
In a distributed setting, using directly those transformations can be perilous, as components can be as large as the diameter of the graph. We adapt Kempe components for distributed algorithms by using an additional color to cut these Kempe components and then permute locally the colors on these components\footnote{Kempe chains have already been used in the distributed setting, see e.g.~\cite{panconesi1995local}}. This will allow us to color independently some parts of the graph with the target coloring. 

Our second main tool is a recoloring technique introduced in~\cite{BousquetB19} for recoloring interval graphs in the centralized setting. We show how to adapt their technique to the distributed setting. The result of~\cite{BousquetB19} essentially ensures that if a large enough (that is, of size $\poly(\Delta)$) number of consecutive vertices $X$ are colored in a desirable way, then we can recolor all the vertices around $X$ with the target coloring by simply recoloring vertices locally. The idea of the proof consists in sliding little by little the set of vertices $X$ colored with the desirable coloring from left to right in such a way that when a vertex leaves the set, it has its target color.

For chordal graphs (Theorem~\ref{thm:recol-chordal}), we use Theorem~\ref{thm:recol-interval} as a black-box, as well as an adaptation of the distributed partition of chordal graphs into interval pieces from~\cite{KonradZ19}. 
We then prove that if we have a recoloring schedule of $G[\cup_{j \ge i+1}V_{j}]$ then we can adapt it into a recoloring schedule for $G[\cup_{j \ge i}V_{j}]$ using Theorem~\ref{thm:recol-interval} and classical recoloring tools.

Notice that the schedule for chordal graphs is large in comparison with the ones for interval graphs (order of $n^{O(\Delta)}$ versus $\poly (\Delta \cdot \logsn)$). This comes from the fact that in the extension of the recoloring schedule of $G[\cup_{j \ge i+1}V_{j}]$ to $G[\cup_{j \ge i}V_{j}]$, we need to leave a polynomial amount of recoloring steps to recolor vertices of $V_i$ between consecutive recolorings of vertices of $G[\cup_{j \ge i+1}V_{j}]$.

The schedule having a $\log n$ factor instead of a $\logsn$ factor originates once again from the decomposition of~\cite{KonradZ19} that uses a logarithmic number of layers. 

\subsection{Organization of the paper}

We start with the related work in Section~\ref{sec:related-work}, and some preliminaries in  Section~\ref{sec:preliminaries}.
Then in Section~\ref{sec:decompo-coloring-intervals}, we explain how to decompose interval graphs and how to modify colorings in such graphs. 
As a corollary, we get our coloring result for interval graphs. Section~\ref{sec:recoloring-intervals} is devoted to the recoloring of interval graphs, based on the decomposition. Finally, Section~\ref{sec:chordal} tackles chordal graphs, describing our decomposition, coloring and recoloring results. The tools introduced in details in the first sections are used as building blocks for the last section.

\section{Related work}
\label{sec:related-work}

\paragraph*{Distributed coloring}
Distributed coloring in bounded-degree and general  graphs has been extensively studied. We refer to the monograph \cite{BarenboimE13} for a general overview of the domain. 
In the LOCAL model, the classic coloring problem is $(\Delta+1)$-coloring, where $\Delta$ is the maximum degree of the graph. In this paper, we are interested in coloring and recoloring with a nearly optimal number of colors (which can be much lower than $\Delta$). There are ways to cope with such small color palette while preserving some locality: studying (multiplicative) approximation of the optimal coloring  \cite{Barenboim12,BarenboimEG18} and/or restricting to special graph classes, for example planar graphs \cite{GoldbergPS88, AboulkerBBE19,ChechikM19} or bounded arboricity graphs \cite{BarenboimE10, GhaffariL17}.

The line of work that is the closest to ours is about coloring interval and chordal graphs. It started with~\cite{HalldorssonK14} who proved an 8-approximation of the optimal coloring, in time $O(\logsn)$ in interval graphs. 
This result was then improved to $(1+\epsilon)$-approximation in time $O\left(\frac{1}{\epsilon}\logsn \right)$ in~\cite{HalldorssonK20}. The $\epsilon$ of this theorem can be as small as $2/\omega$, which means that \cite{HalldorssonK20} obtains an $(\omega+2)$-coloring in time $O(\omega \logsn)$.
Coloring in chordal graphs was considered in \cite{KonradZ19}. The authors prove a $(1+\epsilon)$-approximation in time $O\left(\frac{1}{\epsilon}\log n \right)$. Again, the bounds on $\epsilon$ allow to derive an $(\omega+2)$-coloring, and here the running time is $O(\omega \log n)$ rounds.

\paragraph*{Distributed Reconfiguration}

The introduction of recoloring, \emph{i.e.} reconfiguration of colorings, in the distributed setting is due to~\cite{bonamy2018distributed}. They, for instance, provide algorithms to recolor trees and subcubic graphs with one extra color, and an impossibility result for 3-colored toroidal grids with only one extra color. In particular, they find a constant size schedule with $O(\log n)$ communication rounds on trees if an extra color is allowed. Trees are chordal graphs with $\omega=2$, hence our results relate directly to that, and we use a similar \textit{Rake-and-Compress} approach. Then \cite{censor2020distributed} considered the distributed reconfiguration of maximal independent sets (MIS). 
Before these papers, some results of \cite{panconesi1995local} can be seen as recoloring, as they describe a way to find a $\Delta$-coloring from a given $(\Delta+1)$-coloring by using "augmenting paths" that actually are Kempe changes. 

\paragraph*{Centralized graph recoloring}
In the centralized setting, graph recoloring received a considerable attention in the last decades. 

In this overview, we will focus on single-vertex reconfiguration, that is, the model where exactly one vertex is recolored at each step. In that model, it is known that every $k$-coloring can be transformed into any other as long as $k$ is at least the degeneracy $d$ of $G$ plus two~\cite{dyer2006randomly,Cereceda}. However, the number of recoloring steps is \textit{a priori} exponential in $n$. Cereceda conjectured in~\cite{Cereceda} that, when $k \ge d+2$, there exists a quadratic transformation between any pair of $k$-colorings of any $d$-degenerate graph. Bousquet and Heinrich~\cite{BousquetH19+} proved that there always exist $O(n^{d+1})$ transformation between any pair of $k$-colorings when $k \ge d+2$. 

Since chordal graphs are perfect graphs, it is well known that the degeneracy is related to the cliques number by:  $d=\omega-1$.
For interval and chordal graphs, Bonamy et al.~\cite{Bonamy0LPP14} proved that there exists a quadratic transformation between any pair of vertices when $k \ge \omega+1$ and they provide a quadratic lower bound when $k=\omega+1$. This existence of a quadratic transformation has been extended to bounded treewidth graphs~\cite{BonamyB18}. 
Recently Bartier and Bousquet proved in~\cite{BousquetB19} that, when $k \ge \omega+3$, there always exists a linear transformation between any two colorings of a chordal graph. Moreover, in contrast to most of the centralized recoloring algorithms, this algorithm is ``local'' and will be used as one of the blocks of our proof.

\paragraph*{Graph decompositions}

The first phase of our algorithms consists in decomposing the graph into components of small diameter with some additional properties. 
These decompositions are close to what is known as network decompositions, which are partitions of the nodes of the graph into few classes, such that every class has all its connected components of small diameter (see \cite{Ghaffari20} for an overview of this topic, and in particular the recent advances).
For geometric graphs, such as interval, unit-disks graphs, it is known that a network decomposition with a constant number of classes and constant diameter can be obtained in $O(\logsn)$ rounds \cite{KuhnMW05, SchneiderW10}, which implies that all the classic problems such as $(\Delta+1)$-coloring or maximal matching can be solved in $O(\logsn)$ rounds. For trees, \cite{bonamy2018distributed} establishes a network decomposition with constant size components in time $O(\log n)$, which is enough for recoloring.
For our purpose, which is to (re)color interval and chordal graphs with few colors, standard network decompositions are not powerful enough, and we need to build more constrained decompositions. Our decompositions are inspired by the ones of \cite{HalldorssonK20} and~\cite{KonradZ19}. 

\section{Basic properties of interval and chordal graphs}
\label{sec:preliminaries}

\emph{Interval graphs} are intersection graphs of intervals of the real line. 
\emph{Proper interval graphs} are the intersection graphs of a set of intervals of size one. 
Equivalently, they correspond to interval graphs where no interval is included in the other.
It is easy to check that, starting from an interval graph, and keeping only the intervals that are maximal for inclusion, we get a proper interval graph.

An interval graph can also be characterized as an intersection graph of subpaths of a path, and we can extend this definition to define chordal graphs. 
A graph is \emph{chordal} if it is an intersection graph of subtrees of a tree. Equivalently, it is the class of graphs for which all the cycles of length at least $4$ has a chord.

Before we list some properties of interval, proper intervals and chordal graphs, let us define the notion of degeneracy.

\begin{definition}
\label{def:degenerate}
A graph $G$ is \emph{$d$-degenerate} if there exists an ordering $v_1,\ldots,v_n$ of $V$ such that for every $i \le n-1$, the number of neighbors of $v_i$ in $\{ v_{i+1},\ldots,v_n \}$ is at most $d$.
\end{definition}

Let us remind a few facts about interval graphs.

\begin{observation}
\label{prop:interval-graphs} The following holds.
\begin{enumerate}
    \item\label{rem1} Any chordal graph can be colored with $\omega$ colors and is $(\omega-1)$-degenerate \cite{Maffray03}.
    \item The max degree $\Delta$ can be as large as $n$, even if $\omega$ is small (consider $n-1$ disjoint small intervals contained in a big one).
    \item In an interval graph, the set of intervals that are minimal for inclusion corresponds to a proper interval. The same holds for "maximal for inclusion" \cite{HalldorssonK20}.
    \item\label{rem2} In any proper interval graph, $\Delta \leq 2\omega-2$ (since every neighboring interval should cross one of the extremities of the vertex interval.)
\end{enumerate}
\end{observation}

\paragraph*{Clique trees, clique paths and borders}

An essential tool to study chordal graphs is the notion of clique tree.

\begin{definition}{(See \emph{e.g.} \cite{BlairP93}.)}
Let $G$ be a chordal graph. A \emph{clique tree} of~$G$ is a tree $T$ together with a function that associates to every vertex a connected subtree of $T$. 
Two vertices $x$ and $y$ are connected in $G$ if and only if the subtrees of $x$ and $y$ intersect. For each node $u\in T$, the subset of nodes in $G$ whose subtree contains $u$ is called the \emph{bag} of $u$. Note that each bag of a node of $T$ forms a clique in $G$, so each bag contains at most $\omega$ vertices (and there always exists a clique tree on at most $n$ nodes).

For interval graphs, the tree $T$ is a path, and it is called the \emph{clique path} of~$G$ (see Figure~\ref{fig:clique-path} for an illustration).
\end{definition}

We introduce two more notions.
Let $G$ be an interval graph. A set of vertices $X$ is said to be \emph{consecutive} if it consists in the union of the vertices contained in consecutive cliques of the clique path of $G$. The \emph{border} of $X$ is the subset of $X$ connected to at least one vertex which is not in $X$. In other words, it is the set of vertices of the first and last clique of the clique path containing $X$ that have a neighbor not in $X$. One can easily notice that all the vertices of the border of $X$ can be partitioned into two cliques (the ones that belong to the first and the last clique of the set of bags). 

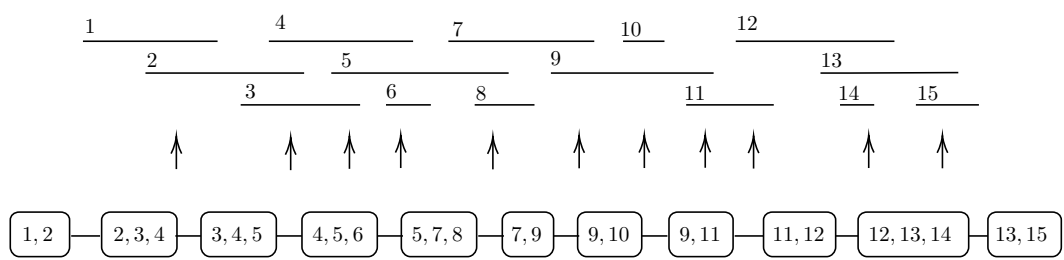
\begin{figure}[!h]
    \centering
    \scalebox{0.8}{
    \tikzset{every picture/.style={line width=0.75pt}} 

\begin{tikzpicture}[x=0.75pt,y=0.75pt,yscale=-1,xscale=1]

\draw    (21,80) -- (105,80) ;
\draw    (60,100) -- (159,100) ;
\draw     (137,80) -- (227,80) ;
\draw    (249,80) -- (340.04,80) ;
\draw    (358,80) -- (383.84,80) ;
\draw     (428.33,80) -- (527.33,80) ;
\draw    (175.94,100) -- (286.61,100) ;
\draw    (312.98,100) -- (414.51,100) ;
\draw    (481.18,100) -- (499.18,100) -- (567.18,99.71) ;
\draw    (119.41,120) -- (193.94,120) ;
\draw    (210.08,120) -- (238.08,120) ;
\draw    (265.41,120) -- (302.74,120) ;
\draw    (397.38,120) -- (452.04,120) ;
\draw    (493.38,120) -- (514.71,120) ;
\draw    (540.71,120) -- (580.04,120) ;

\draw (122,135) node [anchor=north west][inner sep=0.75pt]   [align=left] {$ $};
\draw (122,150.4) node [anchor=north west][inner sep=0.75pt]    {$$};
\draw (20,64) node [anchor=north west][inner sep=0.75pt]   [align=left] {};
\draw (20.67,65) node [anchor=north west][inner sep=0.75pt]    {$1$};
\draw (59.33,86) node [anchor=north west][inner sep=0.75pt]    {$2$};
\draw (120.67,106) node [anchor=north west][inner sep=0.75pt]    {$3$};
\draw (139.33,62.4) node [anchor=north west][inner sep=0.75pt]    {$4$};
\draw (180.67,86) node [anchor=north west][inner sep=0.75pt]    {$5$};
\draw (208.67,106) node [anchor=north west][inner sep=0.75pt]    {$6$};
\draw (250,67) node [anchor=north west][inner sep=0.75pt]    {$7$};
\draw (264.67,107) node [anchor=north west][inner sep=0.75pt]    {$8$};
\draw (311.33,85) node [anchor=north west][inner sep=0.75pt]    {$9$};
\draw (354.67,67) node [anchor=north west][inner sep=0.75pt]    {$10$};
\draw (394.67,107) node [anchor=north west][inner sep=0.75pt]    {$11$};
\draw (427.33,63) node [anchor=north west][inner sep=0.75pt]    {$12$};
\draw (479.83,86) node [anchor=north west][inner sep=0.75pt]    {$13$};
\draw (490.67,107) node [anchor=north west][inner sep=0.75pt]    {$14$};
\draw (540,107) node [anchor=north west][inner sep=0.75pt]    {$15$};

\draw    (79.28,160.38) -- (79.28,141.71) ;
\draw [shift={(79.28,139.71)}, rotate = 450] [color={rgb, 255:red, 0; green, 0; blue, 0 }  ][line width=0.75]    (10.93,-3.29) .. controls (6.95,-1.4) and (3.31,-0.3) .. (0,0) .. controls (3.31,0.3) and (6.95,1.4) .. (10.93,3.29)   ;
\draw    (150.61,161.04) -- (150.61,142.38) ;
\draw [shift={(150.61,140.38)}, rotate = 450] [color={rgb, 255:red, 0; green, 0; blue, 0 }  ][line width=0.75]    (10.93,-3.29) .. controls (6.95,-1.4) and (3.31,-0.3) .. (0,0) .. controls (3.31,0.3) and (6.95,1.4) .. (10.93,3.29)   ;
\draw    (187.28,160.21) -- (187.28,141.54) ;
\draw [shift={(187.28,139.54)}, rotate = 450] [color={rgb, 255:red, 0; green, 0; blue, 0 }  ][line width=0.75]    (10.93,-3.29) .. controls (6.95,-1.4) and (3.31,-0.3) .. (0,0) .. controls (3.31,0.3) and (6.95,1.4) .. (10.93,3.29)   ;
\draw    (219.28,159.71) -- (219.28,141.04) ;
\draw [shift={(219.28,139.04)}, rotate = 450] [color={rgb, 255:red, 0; green, 0; blue, 0 }  ][line width=0.75]    (10.93,-3.29) .. controls (6.95,-1.4) and (3.31,-0.3) .. (0,0) .. controls (3.31,0.3) and (6.95,1.4) .. (10.93,3.29)   ;
\draw    (276.78,160.54) -- (276.78,141.88) ;
\draw [shift={(276.78,139.88)}, rotate = 450] [color={rgb, 255:red, 0; green, 0; blue, 0 }  ][line width=0.75]    (10.93,-3.29) .. controls (6.95,-1.4) and (3.31,-0.3) .. (0,0) .. controls (3.31,0.3) and (6.95,1.4) .. (10.93,3.29)   ;
\draw    (330.61,159.88) -- (330.61,141.21) ;
\draw [shift={(330.61,139.21)}, rotate = 450] [color={rgb, 255:red, 0; green, 0; blue, 0 }  ][line width=0.75]    (10.93,-3.29) .. controls (6.95,-1.4) and (3.31,-0.3) .. (0,0) .. controls (3.31,0.3) and (6.95,1.4) .. (10.93,3.29)   ;
\draw    (371.28,159.71) -- (371.28,141.04) ;
\draw [shift={(371.28,139.04)}, rotate = 450] [color={rgb, 255:red, 0; green, 0; blue, 0 }  ][line width=0.75]    (10.93,-3.29) .. controls (6.95,-1.4) and (3.31,-0.3) .. (0,0) .. controls (3.31,0.3) and (6.95,1.4) .. (10.93,3.29)   ;
\draw    (409.28,159.54) -- (409.28,150.95) -- (409.28,140.88) ;
\draw [shift={(409.28,138.88)}, rotate = 450] [color={rgb, 255:red, 0; green, 0; blue, 0 }  ][line width=0.75]    (10.93,-3.29) .. controls (6.95,-1.4) and (3.31,-0.3) .. (0,0) .. controls (3.31,0.3) and (6.95,1.4) .. (10.93,3.29)   ;
\draw    (439.44,160.04) -- (439.44,141.38) ;
\draw [shift={(439.44,139.38)}, rotate = 450] [color={rgb, 255:red, 0; green, 0; blue, 0 }  ][line width=0.75]    (10.93,-3.29) .. controls (6.95,-1.4) and (3.31,-0.3) .. (0,0) .. controls (3.31,0.3) and (6.95,1.4) .. (10.93,3.29)   ;
\draw    (511.28,159.71) -- (511.28,150.45) -- (511.28,141.04) ;
\draw [shift={(511.28,139.04)}, rotate = 450] [color={rgb, 255:red, 0; green, 0; blue, 0 }  ][line width=0.75]    (10.93,-3.29) .. controls (6.95,-1.4) and (3.31,-0.3) .. (0,0) .. controls (3.31,0.3) and (6.95,1.4) .. (10.93,3.29)   ;
\draw    (557.28,159.88) -- (557.28,141.21) ;
\draw [shift={(557.28,139.21)}, rotate = 450] [color={rgb, 255:red, 0; green, 0; blue, 0 }  ][line width=0.75]    (10.93,-3.29) .. controls (6.95,-1.4) and (3.31,-0.3) .. (0,0) .. controls (3.31,0.3) and (6.95,1.4) .. (10.93,3.29)   ;

\end{tikzpicture}}
    \vspace{0.5cm}
    
    \scalebox{0.80}{
    \tikzset{every picture/.style={line width=0.75pt}} 

\begin{tikzpicture}[x=0.75pt,y=0.75pt,yscale=-1,xscale=1]

\draw   (3,100) .. controls (3,97) and (5.51,95) .. (8.6,95) -- (34.4,95) .. controls (37.49,95) and (40,97) .. (40,100) -- (40,117) .. controls (40,120) and (37.49,123) .. (34.4,123) -- (8.6,123) .. controls (5.51,123) and (3,120) .. (3,117) -- cycle ;
\draw   (60,100) .. controls (60,97) and (62.51,95) .. (65.6,95) -- (101.4,95) .. controls (104.49,95) and (107,97) .. (107,100) -- (107,117) .. controls (107,120) and (104.49,123) .. (101.4,123) -- (65.6,123) .. controls (62.51,123) and (60,120) .. (60,117) -- cycle ;
\draw   (122,100) .. controls (122,97) and (124.51,95) .. (127.6,95) -- (163.4,95) .. controls (166.49,95) and (169,97) .. (169,100) -- (169,117) .. controls (169,120) and (166.49,123) .. (163.4,123) -- (127.6,123) .. controls (124.51,123) and (122,120) .. (122,117) -- cycle ;
\draw   (185,100) .. controls (185,97) and (187.51,95) .. (190.6,95) -- (226.4,95) .. controls (229.49,95) and (232,97) .. (232,100) -- (232,117) .. controls (232,120) and (229.49,123) .. (226.4,123) -- (190.6,123) .. controls (187.51,123) and (185,120) .. (185,117) -- cycle ;
\draw   (247,100) .. controls (247,97) and (249.51,95) .. (252.6,95) -- (288.4,95) .. controls (291.49,95) and (294,97) .. (294,100) -- (294,117) .. controls (294,120) and (291.49,123) .. (288.4,123) -- (252.6,123) .. controls (249.51,123) and (247,120) .. (247,117) -- cycle ;
\draw   (310,100) .. controls (310,97) and (312.51,95) .. (315.6,95) -- (336.4,95) .. controls (339.49,95) and (342,97) .. (342,100) -- (342,117) .. controls (342,120) and (339.49,123) .. (336.4,123) -- (315.6,123) .. controls (312.51,123) and (310,120) .. (310,117) -- cycle ;
\draw   (357,100) .. controls (357,97) and (359.51,95) .. (362.6,95) -- (391.4,95) .. controls (394.49,95) and (397,97) .. (397,100) -- (397,117) .. controls (397,120) and (394.49,123) .. (391.4,123) -- (362.6,123) .. controls (359.51,123) and (357,120) .. (357,117) -- cycle ;
\draw   (414,100) .. controls (414,97) and (416.51,95) .. (419.6,95) -- (448.4,95) .. controls (451.49,95) and (454,97) .. (454,100) -- (454,117) .. controls (454,120) and (451.49,123) .. (448.4,123) -- (419.6,123) .. controls (416.51,123) and (414,120) .. (414,117) -- cycle ;
\draw   (473,100) .. controls (473,97) and (475.51,95) .. (478.6,95) -- (512.4,95) .. controls (515.49,95) and (518,97) .. (518,100) -- (518,117) .. controls (518,120) and (515.49,123) .. (512.4,123) -- (478.6,123) .. controls (475.51,123) and (473,120) .. (473,117) -- cycle ;
\draw   (532,100) .. controls (532,97) and (534.51,95) .. (537.6,95) -- (595.4,95) .. controls (598.49,95) and (601,97) .. (601,100) -- (601,117) .. controls (601,120) and (598.49,123) .. (595.4,123) -- (537.6,123) .. controls (534.51,123) and (532,120) .. (532,117) -- cycle ;
\draw   (613,100) .. controls (613,97) and (615.51,95) .. (618.6,95) -- (653.2,95) .. controls (656.29,95) and (658.8,97) .. (658.8,100) -- (658.8,117) .. controls (658.8,120) and (656.29,123) .. (653.2,123) -- (618.6,123) .. controls (615.51,123) and (613,120) .. (613,117) -- cycle ;
\draw    (41,110) -- (59,110) ;
\draw    (106,110) -- (122,110) ;
\draw    (169,110) -- (185,110) ;
\draw    (232,110) -- (248,110) ;
\draw    (295,110) -- (311,110) ;
\draw    (342,110) -- (358,110) ;
\draw    (397,110) -- (413,110) ;
\draw    (455,110) -- (473,110) ;
\draw    (518,110) -- (532,110) ;
\draw    (601,110) -- (613,110) ;

\draw (9,103) node [anchor=north west][inner sep=0.75pt]    {$1,2$};
\draw (66,103) node [anchor=north west][inner sep=0.75pt]    {$2,3,4$};
\draw (127,103) node [anchor=north west][inner sep=0.75pt]    {$3,4,5$};
\draw (190,103) node [anchor=north west][inner sep=0.75pt]    {$4,5,6$};
\draw (252,103) node [anchor=north west][inner sep=0.75pt]    {$5,7,8$};
\draw (314,103) node [anchor=north west][inner sep=0.75pt]    {$7,9$};
\draw (362,103) node [anchor=north west][inner sep=0.75pt]    {$9,10$};
\draw (419,103) node [anchor=north west][inner sep=0.75pt]    {$9,11$};
\draw (477,103) node [anchor=north west][inner sep=0.75pt]    {$11,12$};
\draw (537,103) node [anchor=north west][inner sep=0.75pt]    {$12,13,14$};
\draw (616,103) node [anchor=north west][inner sep=0.75pt]    {$13,15$};

\end{tikzpicture}}
    \caption{
    The first picture represents a set of intervals.
    The arrows correspond to points of the axis where there is a maximal clique. 
    The second picture describes the clique path of the graph, whose nodes are the maximal cliques.}
    \label{fig:clique-path}
\end{figure}

We will need the following remarks about interval graphs:
\begin{remark}\label{rmk:separatorinterval}
Let $G$ be an interval graph given with an interval representation of $G$. Let $x$ be a vertex of $G$. Then:
\begin{itemize}
    \item $N(x)$ contain all the vertices $v$ whose interval intersects the interval of $x$.
    \item The removal of $N(x) \cup \{ x \}$ separates the vertices with interval at the left of $x$ with the vertices with interval  at the right of $x$.\footnote{Note that there might be more components since, for instance, the "left graph" might not be connected.} A direct consequence is that, for each node $x$ that is not contained in a bag of a node of one extremity of the clique path, its box separates the graph.
    \item For every $r \ge 2$, consider the set  $Y=\cup_{i\le r}N^i(x)$. The subset of vertices of $Y$ that are incident to any vertex of $V \setminus Y$ is composed of at most two cliques $A,B$. Moreover, there is a clique path of $G[Y]$ where $A$ is the first clique and $B$ is the last clique.
\end{itemize}
\end{remark}

\section{Decomposing and coloring interval graphs}
\label{sec:decompo-coloring-intervals}

In the sequential setting, it is easy to compute an interval representation of the graph, and such a representation is often useful in the design of algorithms. 
For example, one can compute a maximum independent set in a greedy manner, scanning the interval by the increasing right-most endpoint.
In our setting, the nodes do not have access to such a representation, and it is not possible to compute one locally. 
We will build a weaker local version of the interval representation to facilitate the design of coloring and recoloring algorithms. This idea originates from \cite{HalldorssonK20, HalldorssonK14}.

We now introduce the notion of \emph{boxes}.
Remember that an $(a, b)$-ruling-set is a set of nodes~$S$, such that for any $u,v \in S$ $u$ and $v$ are at distance at least $a$, and any node is at distance at most $b$ from a node of $S$.

\begin{definition}
Given a $(4, 5)$-ruling-set $S$, for every node $v$ of $S$, we define its \emph{box} as its closed neighborhood in the graph.

\end{definition}

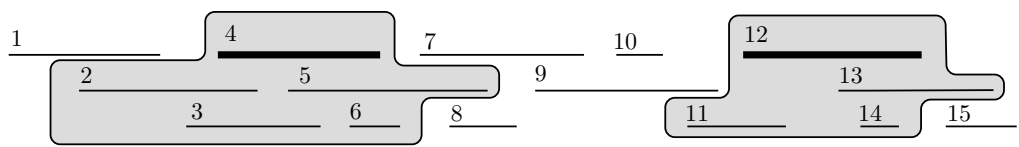
\begin{figure}[!h]
    \centering
    \scalebox{0.9}{
    \tikzset{every picture/.style={line width=0.75pt}} 

\begin{tikzpicture}[x=0.75pt,y=0.75pt,yscale=-1,xscale=1]

\draw  [rounded corners, fill={rgb, 255:red, 220; green, 220; blue, 220 }  ,fill opacity=1 ] (420.51,58) -- (489.84,58) -- (540.51,58) -- (541.18,91) -- (573,91) -- (573,105) -- (527,105) -- (527,126) -- (385.14,126) -- (385.14,104) -- (420.51,104) -- cycle ;
\draw  [rounded corners, fill={rgb, 255:red, 220; green, 220; blue, 220 }  ,fill opacity=1 ] (130,56) -- (197.87,56) -- (235,56) -- (235,86) -- (293,86) -- (293,104) -- (250,104) -- (250,130) -- (44.18,130) -- (44.18,83) -- (130,83) -- cycle ;
\draw    (21,80) -- (105,80) ;
\draw    (60,100) -- (159,100) ;
\draw [line width=3]    (137,80) -- (227,80) ;
\draw    (249,80) -- (340.04,80) ;
\draw    (358,80) -- (383.84,80) ;
\draw [line width=3]    (428.33,80) -- (527.33,80) ;
\draw    (175.94,100) -- (286.61,100) ;
\draw    (312.98,100) -- (414.51,100) ;
\draw    (481.18,100) -- (499.18,100) -- (567.18,99.71) ;
\draw    (119.41,120) -- (193.94,120) ;
\draw    (210.08,120) -- (238.08,120) ;
\draw    (265.41,120) -- (302.74,120) ;
\draw    (397.38,120) -- (452.04,120) ;
\draw    (493.38,120) -- (514.71,120) ;
\draw    (540.71,120) -- (580.04,120) ;

\draw (122,135) node [anchor=north west][inner sep=0.75pt]   [align=left] {$ $};
\draw (122,150.4) node [anchor=north west][inner sep=0.75pt]    {$$};
\draw (20,64) node [anchor=north west][inner sep=0.75pt]   [align=left] {};
\draw (20.67,65) node [anchor=north west][inner sep=0.75pt]    {$1$};
\draw (59.33,86) node [anchor=north west][inner sep=0.75pt]    {$2$};
\draw (120.67,106) node [anchor=north west][inner sep=0.75pt]    {$3$};
\draw (139.33,62.4) node [anchor=north west][inner sep=0.75pt]    {$4$};
\draw (180.67,86) node [anchor=north west][inner sep=0.75pt]    {$5$};
\draw (208.67,106) node [anchor=north west][inner sep=0.75pt]    {$6$};
\draw (250,67) node [anchor=north west][inner sep=0.75pt]    {$7$};
\draw (264.67,107) node [anchor=north west][inner sep=0.75pt]    {$8$};
\draw (311.33,85) node [anchor=north west][inner sep=0.75pt]    {$9$};
\draw (354.67,67) node [anchor=north west][inner sep=0.75pt]    {$10$};
\draw (394.67,107) node [anchor=north west][inner sep=0.75pt]    {$11$};
\draw (427.33,63) node [anchor=north west][inner sep=0.75pt]    {$12$};
\draw (479.83,86) node [anchor=north west][inner sep=0.75pt]    {$13$};
\draw (490.67,107) node [anchor=north west][inner sep=0.75pt]    {$14$};
\draw (540,107) node [anchor=north west][inner sep=0.75pt]    {$15$};

\end{tikzpicture}}
    \caption{An example of a box decomposition. The boxes are the gray areas. The other areas are interboxes. The nodes of the ruling sets are the bold intervals. The intervals 6, 8, 10 and 14, are included into other interval thus they are not considered for the computation of the ruling set.}
    \label{fig:boxes}
\end{figure}

\begin{restatable}{proposition}{propbasicbox}
\label{prop:basicbox}
For a $(4, 5)$-ruling-set $S$ in an interval graph $G$, the following holds:
\begin{enumerate}
\item The boxes of the nodes in $S$ are disjoint, and two nodes from two different boxes cannot be adjacent.
\item The removal of all the boxes leaves a set of connected components, that we call \emph{interboxes}. We have a path alternating boxes and interboxes (by defining adjacency between boxes and interboxes by the existence of an edge linking the two). We call the \emph{virtual path} the path on vertex set $S$ whose adjacency is given by the sequence of boxes defined above.
\item The interboxes have diameter at most 11.
\end{enumerate}
\end{restatable}

\begin{proof}
\begin{itemize}
    \item As the nodes of the $(4,5)$-ruling-set are at distance at least 4, their neighborhoods do not intersect.
    \item By Remark~\ref{rmk:separatorinterval}, each box is a separator of the graph (except for vertices incident to the vertices with the smallest right extremity or with the largest left extremity, that is vertices of the first and last bag by Remark~\ref{rmk:separatorinterval}). So, since boxes are disjoint by the first point, for any $a,b,c \in S$, there exists one of them, w.l.o.g. $a$, such that any $bc$-path contains a vertex in the neighborhood of $a$. 
    It permits to define a \emph{virtual path} where the neighbors of $a \in M$ are the (at most two) vertices $x$ of $M$ such that there is a path from $a$ to a neighbor $y$ of $x$ which does not intersect $N(z)$ for $z \in M \setminus \{a,x\}$.

    \item Consider an interbox between two boxes centered at nodes $s_1$ and $s_2$ of the ruling set. By definition of the $(4,5)$-ruling set, the shortest path between $s_1$ and $s_2$  has length at most 11. The part of this path that is in the interbox has length at most 9. Now, because of the structure of interval graphs, every node of the interbox must be adjacent to at least one of the nodes of this inner part of the path. Therefore, all the nodes of the interbox are at distance at most 11 from each other. 
\end{itemize}
\end{proof}

Now from a distributed point of view, we say that the nodes compute a box decomposition if they compute a $(4, 5)$-ruling-set $S$, and for every node, either it is in a box, and then it knows the two adjacent boxes, and the structure of the graph between these boxes, or it is not in a box, and knows between which two boxes it is, and the structure of the graph between these boxes. 

Let $S$ be a $(4,5)$-ruling set of $G$. If we replace every box a by a single vertex and every interbox by an edge, the resulting graph is a path called the \emph{ruling path}. Two vertices of $S$ are said to be consecutive if they are adjacent in the ruling path, and at distance at most $r$ if they are at distance at most $r$ in the ruling path.

\begin{restatable}{lemma}{LemComputeBoxDecomposition}
\label{lem:compute-box-decomposition}
A box decomposition can be found in $O(\logsn)$ rounds in the LOCAL model.
\end{restatable}

\begin{proof}
Once a $(4, 5)$-ruling-set $S$ is computed, every node can just look at constant distance to find out the box structure around itself.

In order to compute the $(4, 5)$-ruling-set $S$, we basically follow the line of reasoning of \cite{HalldorssonK20}. 
First, the nodes can, in constant time, decide whether their intervals are maximal for inclusion (the node just needs to check if its neighborhood is included in one of its neighbors' neighborhood). Let us denote by $W$ these vertices. Then, because the subgraph $G'$ of $G$ induced by $W$ is a proper interval graph, one can use the algorithm of \cite{SchneiderW10} and get a maximal independent set $M$ of $G'$ in time $O(\logsn)$.\footnote{Paper \cite{SchneiderW10} gives an algorithm working in graph classes that have bounded independence, and proper intervals is such a class.} 
This maximal independent set is a $(2, 2)$-ruling set of the original graph $G$. 
Indeed, the set $M$ is by definition a $(2,1)$-ruling-set of $G'$, and every node of $G$ is at distance at most one from a node of $G'$.

Notice that none of those intervals overlap, since $M$ is an independent set. Hence, we can order them according to their increasing extremities of the corresponding intervals, creating a new path (two intervals are connected if they follow directly one another). Those nodes can simulate this path in a distributed manner. 
We simply compute an MIS in this path. The nodes of this new MIS $S$ are at distance at least $4$ in the original graph $G$, and any node is at distance at most $5$ from a node of $S$.
\end{proof}

Let $M$ be a box decomposition. As in the proof of Lemma~\ref{lem:compute-box-decomposition}, there is a natural virtual path between these boxes. 
Let $B,B'$ be two boxes of a box decomposition. We say that the boxes $B,B'$ are at distance $r$, if they are at distance $r$ in the virtual path of the box decomposition.
The subgraph between $B$ and $B'$ is the subgraph of $G$ composed of the boxes $B,B'$, all the boxes $B$ and $B'$ (in the virtual path) and all the interboxes between $B$ and $B'$. Note that if the boxes $B$ and $B'$ are at distance $r$, then the subgraph between $B$ and $B'$ can be computed in $O(r)$ rounds in the LOCAL model.

We will now prove a technical lemma that will be helpful in several places of the paper. 

\begin{restatable}{lemma}{LemSwitch}
\label{lem:switch}
Let $A,B$ be two boxes of two nodes at distance $3$ and $H$ be the subgraph between $A$ and $B$.
Let $\alpha$ be a proper $k$-coloring  of $H$, and let $x$ and $y$ be two arbitrary colors. 
Let $\beta$ be the $k$-coloring of $H$ made by permuting $x$ and $y$ in $\alpha$.
There exists a $(k+1)$-coloring of $H$, that corresponds to $\alpha$ on $A$ and $\beta$ on $B$.
\end{restatable}

\begin{proof}
{Let $B_0=A$, $B_1$, $B_2$ and $B_3=B$ be the boxes between $A$ and $B$.}
We denote by $X_{i,i+1}$ the interbox between the boxes $B_i$ and $B_{i+1}$.
We start from the coloring $\alpha$, and modify it several times, until we reach a coloring satisfying the condition of the lemma. 

First, consider a coloring $\alpha_1$ that coincides with $\alpha$ on the nodes in $A\cup X_{0,1}$ and on the nodes of $H$ that do not have color $x$ in $\alpha$. The nodes that are outside $A\cup X_{0,1}$ and have color $x$ in $\alpha$ are given color $k+1$. Note that we have just recolored an independent set with a new color, therefore $\alpha_1$ is a proper coloring.

Second, consider $\alpha_2$ that coincides with $\alpha_1
$ on the nodes in $A\cup X_{0,1}\cup B_1\cup X_{1,2}$ and on the nodes of $H$ that do not have color $y$ in $\alpha_1$. The nodes that are outside $A\cup X_{0,1}\cup B_1\cup X_{1,2}$ and have color $y$ are given color $x$. Again, the coloring computed is proper, because we have recolored an independent set with a color that was not used in this part of the graph. Indeed, all nodes of color $x$ in $\alpha_1$ are in $A\cup X_{0,1}$, and by definition of the boxes, $X_{1,2}$ forms a cut between those nodes and the newly colored nodes.

Finally, consider $\alpha_3$ that coincides with $\alpha_2$ on the nodes in $A\cup X_{0,1}\cup B_1\cup X_{1,2}\cup B_2\cup X_{2,3}$ and on the nodes of $H$ that do not have color $k+1$ in $\alpha_2$. The nodes that are outside $A\cup X_{0,1}\cup B_1\cup X_{1,2}\cup B_2\cup X_{2,3}$ and have color $k+1$ are given color $y$. Those nodes form an independent set, and do not have color $y$ in their neighborhood. Indeed, all nodes of color $y$ with $\alpha_2$ are in $A\cup X_{0,1}\cup B_1\cup X_{1,2}$, and by definition of the boxes, $X_{2,3}$ forms a cut between those nodes and the newly colored nodes.

The coloring $\alpha_3$ satisfies the conditions of the lemma. 
\end{proof}

\begin{restatable}{lemma}{LemABcoloring}
\label{lem:coloring-interval-border}
Consider the subgraph $H$ induced by the nodes of two boxes $A$ and $B$, separated by at least $3k$ other boxes, and all the nodes in between.
Consider also two arbitrary proper $k$-coloring $\alpha$ and $\beta$ of $H$. 
Then there exists a $(k+1)$-coloring of $H$, that corresponds to $\alpha$ on $A$ and $\beta$ on $B$.
\end{restatable}

\begin{proof}
Let us make a preliminary observation. 
The notion of boxes we have defined is handy in the distributed setting, but for reasoning about coloring, cliques are easier to manipulate. 
Indeed, in a clique, going from one coloring to a second can be seen as a permutation of colors, because each color is used only once. On the other hand, in a box, several intervals can have the same color.
To keep the idea of a color permutation, we will consider $A_c$ (resp. $B_c$), which is the clique of the intervals intersecting the right-most (resp. left-most) endpoint of the center of $A$ (resp. $B$). If we can build a coloring as described in the lemma for $A_c$ and $B_c$, then we can directly extend it to $A$ and $B$.

Now the idea is to transform $\alpha$ into $\beta$ by using $k$ times Lemma \ref{lem:switch}. The permutation of colors from $\alpha$ to $\beta$ in $B_c$ can be done by a sequence of at most $k$ switches of pairs of colors ($x_1$ with $y_1$, $x_2$ with $y_2$, etc.).

Let $B_1$, $B_2$, etc. be the boxes between $A$ and $B$ (with $A_c=B_0$ and $B_c=B_{3k+1}$). We will again denote by $X_{i,i+1}$ the nodes between the box $B_i$ and $B_{i+1}$. We will create iteratively a sequence of coloring as follows:
\begin{enumerate}
    \item $\gamma_1$ is the coloring of $B_3\bigcup\limits_{j\ge3}(X_{j,j+1}\cup B_{j+1})$ where all nodes have the same color than in $\alpha$ except that color $x_1$ and  $y_1$ have been exchanged.
    \item We apply Lemma \ref{lem:switch} on $A$ and $B_3\bigcup\limits_{j\ge3}(X_{j,j+1}\cup B_{j+1})$ to get $\alpha_1$. (Note that $B_3\bigcup\limits_{j\ge3}(X_{j,j+1}\cup B_{j+1})$ is not a box, but the construction of Lemma~\ref{lem:switch} works just the same. Also note that the last box is $B_c$ and not $B$.)
    \item $\forall i\ge2, \gamma_i$ is the coloring of $B_{3i}\bigcup\limits_{j\ge3i}(X_{j,j+1}\cup B_{j+1})$ where all nodes have the same color as in $\alpha_{i-1}$ besides the nodes of color $x_{i}$ that gets color $y_i$ and reversely.
    \item We apply Lemma \ref{lem:switch} on $A\bigcup\limits_{j\le3(i-1)}(X_{j,j+1}\cup B_{j+1})$ and $B_3\bigcup\limits_{j\ge3}(X_{j,j+1}\cup B_{j+1})$ to get $\alpha_{i+1}$.
\end{enumerate}
At the end, we get that the colors of nodes on $A$ have never changed, and that we have applied the sequence of switches through the path meaning that $\alpha_{i+1}$ corresponds to $\beta$ on $B$.
\end{proof}

Given this tool, we easily get the following theorem.

\ThmColInterval*

\begin{proof}
One first computes the box decomposition in time $O(\logsn)$, then, given the paths of boxes, one can iterate an MIS algorithm for paths (\emph{e.g.} Cole-Vishkin algorithm \cite{ColeV86}) to compute a $(3\omega, 6\omega)$-ruling set $S$ of boxes. 
This uses $O(\omega \logsn)$ rounds. 
Then the nodes of the boxes of $S$ computes an $\omega$-coloring of their own box. Finally, we use the lemma above to fill the gaps, which also takes $O(\omega)$ rounds.
\end{proof}

\section{Recoloring interval graphs}
\label{sec:recoloring-intervals}

The goal of this section is to prove the following theorem.

\ThmIntervalle*

\medskip

The result is tight in terms of number of colors for the two last items, because of the following lemma:

\begin{lemma}
\label{lem:lower-bound}
In interval graphs of clique number $\omega$, if no additional color is allowed, then finding a recoloring from a $c$-coloring to a $c'$-coloring with $c,c'<2\omega$ requires $\Omega(n/\omega)$ rounds in the LOCAL model and a schedule of length $\Omega(n/\omega)$.
\end{lemma}

\begin{proof}
We will consider the $\omega$-th power of a path.
We can build an interval representation of such a graph, by representing every vertex at position $i$ by the interval $[i+1/4, i+k+3/4]$.
If we consider a power of a path of clique number $\omega$, and color its $i$-th vertex with color $i \mod 2\omega-1$, all the vertices but the first and the last $\omega$ ones are \emph{frozen} (i.e. we cannot change their color without changing the color of some vertex in its neighborhood before). 
Thus, a recoloring can only happen little by little starting from the extremities of the interval graph. 
Thus, a recoloring schedule has length $\Omega(n)$ and the nodes in the middle of the power path cannot stop before $\Omega(n)$ rounds, as their slot in the schedule will depend on the length of the power path. Moreover, a node in the middle needs to know its distance to an extremity, which takes $\Omega(n)$ rounds in the LOCAL model.
\end{proof}

\subsection{Outline of the proof}

The rest of the section is devoted to prove Theorem~\ref{thm:recol-interval}. In this subsection, we will first explain the shape of the proof, and then we detail the proof in the next sections. Let $\alpha,\beta$ be two $k$-colorings.
\medskip

\noindent
\textbf{Step 1:} Compute a $(4,5)$-ruling set $S$ in $O(\logsn)$ rounds.
\medskip

\noindent
\textbf{Step 2:} In graph recoloring, instead of transforming $\alpha$ into $\beta$, a classical method consists in proving that both $\alpha$ and $\beta$ can be recolored into a so-called canonical coloring $\gamma$ with desirable properties (see e.g.~\cite{Bonamy0LPP14,BousquetB19}...). If transformations $\mathcal{S}_1,\mathcal{S}_2$ from respectively $\alpha$ and $\beta$ to $\gamma$ exist, then $\alpha$ can be transformed into $\beta$ via the transformation $\mathcal{S}_1\mathcal{S}_2^{-1}$. \\
Theorem~\ref{thm:col-interval}
 ensures that an $(\omega+1)$-coloring of $G$ can be found in $O(\omega \logsn)$ rounds. 
 Let us denote by $\gamma$ such a coloring. The goal of the proof will consist in proving that we can find a transformation from $\alpha$ to $\gamma$.
\medskip

\noindent
\textbf{Step 3:} (Section~\ref{sec:reducecolors}). \\
We show that if $k \ge 2\omega$, then we can reduce the number of colors without any additional color. We first compute an independent set $S'$ of $S$ at constant distance. The main idea of this step consists in proving that all the vertices in the interboxes between two consecutive vertices of $S'$ but the vertices of the boxes of $S'$ can be recolored with a color smaller than $2\omega$ without recoloring the boxes of $S'$ (which guarantees that we can perform this recoloring simultaneously everywhere in the graph). By repeating this operation twice, we then prove that all the vertices are recolored with a color smaller than $2\omega$.

When the number is at least $4\omega$, we prove a stronger result, since we directly provide a transformation from $\alpha$ to $\beta$ in $O(\omega\Delta)$ rounds. 
\medskip

\noindent
\textbf{Step 4:} (Section~\ref{sec:kempe}) \\
At this point, after applying Step 3 if $k \ge 2\omega$, we can assume that we have one additional color. Indeed, in the two first cases of Theorem~\ref{thm:recol-interval}, we are allowed to use at least one extra color, and when $k \ge 2\omega$, we have freed at least one color at Step 3.

The goal of Step 4 consists in coloring some boxes with their colors in $\gamma$. To do that, we will perform some Kempe changes, cut at some large enough distance, using one additional color. 
By repeating this process several times, we will prove that a box plus all its neighbors at distance $\Omega(\poly \Delta)$ can be colored with the target coloring.

Using this technique, we can prove that some portions of diameter $\Omega(\poly(\Delta))$ of the interval graph which are at distance $f(\Delta)$ from each other are colored with the target coloring $\gamma$. 
\medskip

\noindent
\textbf{Step 5:} (Section~\ref{sec:buffer}) \\
By Step 4, we can assume that some portions of diameter $\Omega(\poly(\Delta))$ of the interval graph which are at distance $f(\Delta)$ from each other are colored with the target coloring $\gamma$.
We can now use as a black-box a result of Bartier and Bousquet~\cite{BousquetB19} that ensures that we can recolor all the vertices between a consecutive set of boxes with the target coloring without recoloring the border of these sets as long as a sufficiently large number of consecutive vertices of the initial coloring are colored "nicely". Since many consecutive vertices of the initial coloring are colored with the target coloring and such a coloring is "nice", we can use the result of~\cite{BousquetB19} as a blackbox.

After Step 5, all the vertices have received their final color, which concludes the proof.

\subsection{Reducing the number of colors to $2\omega-1$}\label{sec:reducecolors}

Let $G$ be an interval graph and $S$ be a $(4,5)$-ruling set of $G$. Let $x,y \in S$. The graph $G_{x,y}$ between $x$ and $y$ is the graph containing the vertices of the boxes of $x$ and $y$ and all the vertices that are between $x$ and $y$.
Note that since every vertex is at distance at most $5$ from a vertex of the ruling set $S$, two consecutive vertices of a ruling set are at distance at most $10$. 

 Observe that the border of $G_{x,y}$ is a subset of the boxes of $x$ and $y$. 
Moreover, this border contains at most $2\omega$ nodes since otherwise at least $\omega$ neighbors of $x$ or $\omega$ neighbors of $y$ would have neighbors outside $G_{x,y}$, a contradiction with the fact that $G$ has clique number $\omega$. We can also observe that if the graph $G_{x,y}$ between $x$ and $y$ can be found in $t$ rounds, then we can determine the border of $G_{x,y}$ in $t+1$ rounds.

Every interval graph is $\omega$-colorable. Indeed, every interval graph admits a simplicial vertex, that is a vertex whose neighborhood is a clique (and note that this clique must have size at most~$\omega-1$). By peeling interval graphs by removing iteratively simplicial vertices, we can easily obtain an $\omega$-coloring. However, to obtain such an ordering, we need to start at the border of the graph. Indeed, if we consider the $k$-power of a path (which is $k$-degenerate), there are only two degree $k$ vertices which are at the ``borders'' of the graph. If we consider an induced subgraph of bounded diameter in the core of that graph, all the vertices have degree $2k$.

Let $G$ be a graph and $Z$ be a subset of $V$. Remember that $Z$ is \emph{$d$-degenerate} if there exists an ordering $v_1,\ldots,v_r$ of $Z$ such that $N(v_i) \cap ((V \setminus Z) \cup \{ v_{i+1},\ldots,v_r \})$ (that is the neighborhood of $v_i$ where we remove any $v_j$ with $j<i$) is at most $d$ (Definition~\ref{def:degenerate}). In other words, it is possible to find a degeneracy ordering of $V$ starting with the vertices of $Z$ such that, when they are deleted, the vertices of $Z$ have degree at most $d$.
Even if an induced subgraph of an interval graph is not $d$-degenerate, we can prove the following:

\begin{restatable}{lemma}{Lem2omegadegen}
\label{lem:2omegadegen}
Let $G$ be an interval graph. Let $x,y$ be two vertices of a $(4,5)$-ruling set $S$. Let $V'$ be the subset of vertices between $x$ and $y$, $G'$ be the induced subgraph of $G$ by $V'$,  and $X$ be the border of $G'$. The set $V' \setminus X$ is $(2\omega-2)$-degenerate.
\end{restatable}

\begin{proof}
Consider an interval representation of $G'$ where the vertices with leftmost extremities are the vertices of $N(x) \cap X$, and vertices with the rightmost extremities are the vertices of $N(y) \cap X$. Let $v$ be a vertex in $V' \setminus X$ of minimum degree; if several choices are possible, let us choose $v$ of minimum length. We claim that $v$ has degree at most $2\omega-2$. Indeed, every vertex intersecting the interval of $v$ intersects the left or the right endpoint of $v$ since 
if a vertex $w$ is included in the interval of $v$, either it is in $V'$ and we have a contradiction or it is in $X$ but then $v$ also is in $X$ since $N(v)$ contains $N(w)$.
So all the neighbors of $v$ contain the left extremity of $v$ or the right extremity of $v$ in their interval. Since at most $\omega$ vertices pass through each of point of the line (all the vertices containing a point being a clique) and $v$ is one of them, $v$ has at most $2\omega-2$ neighbors.

We can now delete $v$ and repeat until $V' \setminus X$ is empty to obtain the desired degeneracy ordering.
\end{proof}

Note that the function on $\omega$ cannot be improved since it is tight for an $(\omega-1)$-power path.%

By abuse of notation, we say that two vertices $x,y$ of a ruling set $S$ are \emph{at distance $r$} if they are at distance $r$ in the ruling path.

One can then obtain the following corollary which is a simple adaptation of a classical statement of graph recoloring (see e.g.~\cite{Cereceda,dyer2006randomly,BousquetP16}):

\begin{restatable}{corollary}{LemRecolbetweenruling2}
\label{lem:recolbetweenruling2}
Let $G$ be an interval graph and $k \ge 2\omega$. Let $x,y$ be two vertices of a $(4,5)$-ruling set $S$ at distance $r$ in the graph. Let $G'$ be the graph induced by the vertices between the boxes of $x$ and $y$ and $X$ be its border. 
\begin{itemize}
    \item There exists a recoloring schedule of $V'=V(G')\setminus X$ of length at most $2^{|V'|}-1$ such that, at the end of the schedule, no vertex of $V'$ has a color larger than $2\omega-1$.
    \item Moreover, if $k \ge 4\omega$, there exists a recoloring schedule of $V'=V(G')\setminus X$ of length at most $\omega \cdot |V'|$ such that, at the end of the schedule, no vertex of $V'$ has a color larger than $2\omega-1$.
\end{itemize}
\end{restatable}

\begin{proof}
Let us denote by $v_1,\ldots,v_n$ an ordering of $V'=V(G') \setminus X$ such that $v_i$ has degree at most $2\omega-2$ in $G\setminus\{v_1,\ldots,v_{i-1}\}$, which exists since $V'$ is $2\omega-2$-degenerate in $G$ by Lemma~\ref{lem:2omegadegen}.
Let $G_i'$ be the subgraph of $G'$ induced by $\{ v_i,\ldots,v_n \} \cup X$. Note that any coloring of $G[V \setminus V']$ can be extended to $V'$ using only colors in $\{1,\ldots,2\omega-1\}$ since the degeneracy of $V'$ is  $2\omega-2$.

Let us prove the first statement using backward induction on $i$. In $G_n'$ the vertex $v_n$ is only adjacent to vertices in $X$, and $|X| \leq 2\omega-2$, hence in one step, we can recolor it to a color smaller than $2\omega$ and the conclusion follows. Assume now by induction that the conclusion holds for $G_{i+1}'$. We can extend the recoloring procedure of $G_{i+1}'$ to $G_i'$ by simply recoloring $v_i$ each time one of its neighbors has to be recolored with its current color. Since the total number of colors is at least $2\omega$, each time a neighbor of $v_i$ has to be recolored with its actual color, we can first recolor $v_i$ with a different color since its neighborhood has size at most $2\omega-2$ in $G_i'$ by definition of $v_i$. Moreover, the final color of $v_i$ might be $2\omega$, we might need one last step to recolor it with a color smaller than $2\omega$. Hence, the number of recoloring steps for $v_i$ is bounded above by the sum of the numbers of recoloring steps of the nodes $v_{i+1},\ldots, v_n$. By backward induction on $i$, we can notice that this number is at most $2^{|V(G_i')|}-1$.

The proof of the second point consists in applying the same technique, except that we have more than one choice when we recolor the vertices. We can then "look further" to choose better the next vertex to recolor. A complete analysis giving this bound can be found in Theorem~1 of~\cite{BousquetP16}.
\end{proof}

\begin{restatable}{theorem}{Thm2Omega}
\label{Thm2Omega}
Let $G$ be an interval graph, $S$ be a $(4,5)$-ruling set of $G$ and $k \ge 2\omega$ (resp. $4\omega$). 
We can reduce the number of colors from $k$ to $2\omega-1$ with a schedule of length $2^{\mathcal{O}(\Delta)}$ (resp. $O(\omega\Delta)$)
in $O(\logsn)$ rounds in the LOCAL model.
\end{restatable}

\begin{proof}
For every node $v$ in $S$, we denote by $B_v$ the box of $v$. By Proposition~\ref{prop:basicbox}, the distance between two consecutive boxes is at most $8$ and the diameter of a box is at most $2$ (since all the vertices of the box are incident to the corresponding vertex of the ruling set). A maximal independent $I$ at distance $6$ in the ruling path can then be computed in $O(\logsn)$ rounds.

Each vertex of $S$ gets, as a label, its distance to $I$ in the ruling path. This labeling can be done in $O(1)$ rounds, as $S$ nodes are at distance at most 6 to $I$ in the ruling path. We can notice that between two labelled $0$ nodes, there exists (at least one) labelled $3$ node (since $I$ is an independent set of $S$ at distance $6$ in the ruling path). 

Now we perform the following recoloring sequence:
\begin{enumerate}
    \item For every $0$-labelled vertex, we denote by $x$ and $y$ its two closest labelled $2$ vertices of $S$. Let $G'$ be the subgraph between $x$ and $y$. We recolor all the vertices of $G'$ but its border with a color smaller than $2\omega$ using Lemma~\ref{lem:recolbetweenruling2}. We perform this recoloring simultaneously for all the $0$-labelled vertices.
    \item Let $G''$ be the graph between two consecutive $1$-labelled vertices which do not contain $0$ vertices. We recolor all the vertices of $G''$ minus its border with a color smaller than $2\omega$ by Lemma~\ref{lem:recolbetweenruling2}. 
\end{enumerate}
Note that, after these two sequences, all the vertices are recolored. Indeed, the borders of $G'$ and $G''$ are disjoint since the neighborhood of two consecutive vertices of the ruling set are disjoint.

Let us now determine the length of the schedule. To do so, let us bound the number of vertices in each graph $G'$ and $G''$. Since two consecutive vertices of the ruling set are at distance at most $10$, the diameter of $G'$ is at most $50$ (since the labeled $0$ node is at distance $2$ from the two nodes labelled $2$ in the ruling path). Since a connected interval graph always contains a dominating induced path whose length is the diameter of the graph, $G'$ contains at most $50 \Delta$ vertices.

Let us now count the number of vertices between two consecutive $1$. We claim that no label is larger that $5$ since otherwise $I$ would not be maximal. So the number of vertices of $G''$ is at most $100 \Delta$. The length of the schedule is then a consequence of Lemma~\ref{lem:recolbetweenruling2}.
\end{proof}

\begin{corollary}
Let $G$ be an interval graph and $k \ge 4\omega$. Let $\alpha,\beta$ be two $k$-colorings of $G$. Then one can transform $\alpha$ into $\beta$ with a schedule of length $O(\omega\Delta)$ in $O(\logsn)$ rounds in the LOCAL model.
\end{corollary}
\begin{proof}
By Theorem~\ref{Thm2Omega}, we can transform $\alpha$ (resp. $\beta$) into a $2\omega$-coloring $\alpha'$ (resp. $\beta'$) of $G$ with a schedule of length $O(\omega\Delta)$ in $O(\logsn)$ rounds. 
Let us prove that we can transform $\alpha'$ into $\beta'$ in $4\omega$ steps. At step $i \le \omega$, we recolor all the vertices colored with $i$ in $\beta'$ with color $2\omega+i$. Such a set is indeed an independent set since it is a color class in $\beta'$. Note that since $\alpha'$ is a $2\omega$-coloring, no vertex is colored with color $2\omega+i$. After $2\omega$ steps, all the vertices are colored with colors in $\{2\omega+1,\ldots,4\omega\}$. Then, at steps $2\omega+1\le j \le 4\omega$, we color all the vertices colored $j$ with color $j-2\omega$, which completes the proof.
\end{proof}

\subsection{Coloring consecutive cliques with the target coloring}\label{sec:kempe}

\begin{definition}
A \emph{Kempe component} of a coloring $\alpha$ is a bichromatic connected component. 
In other words, $W$ is a Kempe component if there exists two colors $a,b$ such that $W$ is a connected component of the graph restricted to the vertices colored $a$ or $b$ in $\alpha$. A \emph{Kempe change} consists in permuting the colors $a$ and $b$ of a Kempe component. One can easily notice that the resulting coloring is still proper. Kempe changes are a classical tool to design centralized (re)coloring algorithms.

\end{definition} 

The goal of this part is to show that we can color consecutive cliques of the clique path  with the target coloring. The proof will be based on the important lemma below (Lemma~\ref{lem:mainkempe}).

Let $\mathcal{C}$ be a clique path of $G$ (given with an orientation)  such that 
each bag contains exactly one vertex that is not in the previous one (i.e. if $b_i$ and $b_{i+1}$ are two consecutive bags, $|b_{i+1}\setminus b_i|=1$) 
\footnote{Note that such a path decomposition always exist and is called a nice path decomposition. Indeed, if given two consecutive nodes several vertices appear, then we can add intermediate bags to make them appear one after another.} and $X$ be a subset of vertices. 
For each $x\in X$, let $f(x)\in\mathcal{C}$ be the first clique where $x$ appears. We say that the vertices of $X$ are \emph{consecutive with respect to $\mathcal{C}$} if $\cup f(x)$ is connected.
Let $x,y$ be two vertices of a $(4,5)$-ruling set and $X$ be the vertices between $x$ and $y$ without its border.
Note that if $X$ is a connected set of diameter $r$, then any vertex of $X$ can check in the LOCAL model if $X$ is a set of consecutive vertices in $r+1$ rounds.

Given a coloring $\alpha$ and a subset $X$ of vertices, we denote by $\alpha_S$ the coloring of $\alpha$ restricted to the vertices of $S$.

In order to prove the next lemma, we will need the following claim:

\begin{lemma}
Let $G$ be an interval graph, $\alpha,\beta$ be two colorings of $G$, and $S$ be a subset of consecutive vertices. Then, starting from $\alpha$, we can obtain a coloring $\alpha'$ such that $\alpha'_S=\beta_S$, with at most $|S|$ Kempe changes.
\end{lemma}

\begin{proof}
Let $s_1,\ldots,s_\ell$ be an ordering of the intervals of $S$ by increasing starting point. 
We prove the claim by induction on the elements of  $S$, in that order.
For $s_1$, if it does not have color $\beta(s_1)$, we can simply apply the Kempe change on the $(\alpha(s_1),\beta(s_1))$-component containing $s_1$.
Now assume that all the vertices of $S$ up to $s_i$ have been colored with their target color in $\beta$. Let $s_{i+1}$ be the next vertex in $S$. For simplicity, let us still denote the current coloring by $\alpha$. If $\alpha(s_{i+1})=\beta(s_{i+1})$, we do not have to change this color and the conclusion follows. Now assume that $\alpha(s_{i+1})\ne \beta(s_{i+1})$. If we apply a Kempe change on the $(\alpha(s_{i+1}),\beta(s_{i+1}))$-component containing $s_{i+1}$, then $s_{i+1}$ would receive its target color.
We prove, by applying this change, that none of the vertices $s_1,\ldots,s_i$ is recolored (and then these nodes are still colored with their right color). 
If the Kempe chain containing $s_{i+1}$ contains $s_j$ with $j\le i$, then it means the Kempe component contains a vertex $s$ of $S$ adjacent to $s_{i+1}$. Indeed, either $s_j$ is a neighbor of $s_{i+1}$, or there is a path from $s_j$ to $s_i$. All the vertices of this path must belong to $S$ by definition of consecutive vertices. The current color of $s$, which is already colored with its target color, cannot be $\beta(s_{i+1})$ since otherwise the target coloring is not proper and cannot be $\alpha(s_{i+1})$ since otherwise the current coloring is not proper, a contradiction.
\end{proof}

Now, as we know that any subpath $S$ can be recolored with $|S|$ Kempe changes, we can perform those changes on a larger path. Here, we prove how to simulate them locally with an extra color:

\begin{restatable}{lemma}{LemMainkempe}
\label{lem:mainkempe}
Let $G$ be an interval graph with a clique path $\mathcal P$ and let $B$ be its border. Let $S$ be a subset of consecutive vertices, $\alpha$ be a $k-1$ coloring of $G$, and $\beta_S$ be a $k-1$ coloring of $G[S]$. Assume that the distance between $B$ and $S$ is $D \geq 3 |S|$. Then there exists a $k$-coloring $\beta$ of $G$ and a recoloring schedule from $\alpha$ to $\beta$ such that: 
\begin{itemize}
    \item $\beta$ is an extension of $\beta_S$; 
    \item the vertices of $B$ are not recolored in the recoloring sequence (and in particular $\beta_B=\alpha_B$); 
    \item the schedule has length at most $3 |S|$.
\end{itemize}
\end{restatable}

\begin{proof}
Let us denote by $c_0$ the additional color which is not used in the colorings $\alpha$ and $\beta_S$.
We will prove by induction on $A$ the following property: assume that there exists in $G[S]$ a recoloring sequence using Kempe chains from $\alpha|_S$ to $\beta_S$ of length at most $A$, and the distance between $S$ and $B$ is $D \geq 3 A$, then there exists a coloring $\beta$ and a recoloring sequence from $\alpha$ to $\beta$ satisfying the properties of the statement of the Lemma.

When $A = 0$, this property is trivially true by taking $\alpha = \beta$. Hence, we will assume in the following $A > 0$. Let $C_S$ be the first recolored Kempe chain in a shortest recoloring sequence from $\alpha|_S$ to $\beta_S$. Let $v$ be a vertex recolored by this Kempe chain from the color $c_1$ to $c_2$, and let $C$ be the Kempe chain of the whole graph $G$ which recolors $v$ from $c_1$ to $c_2$. Because $S$ is composed of consecutive vertices, and $G$ is an interval graph, we must have $C_S = C \cap S$ since otherwise $G$ would contain an induced cycle of length at least $4$.

Let us consider the following recoloring schedule:
\begin{enumerate}
    \item recolor the vertices of $C$ colored $c_1$ at distance at most D from $S$ with the additional color $c_0$;
    \item recolor the vertices of $C$ colored $c_2$ which are not adjacent to a vertex colored $c_1$ with the color $c_1$ ; 
    \item recolor the vertices colored $c_0$ which do not have the color $c_2$ in their neighborhood with the color $c_2$ .
\end{enumerate}

Let $\alpha'$ be the resulting coloring of $G$. Note that $\alpha'$ is obtained from $\alpha$ with a recoloring schedule of length $3$. Moreover, the vertices at distance at most $D-2$ from $S$ are not colored with $c_0$. Indeed, after the first step, there is no vertex of $C$ colored $c_1$ at distance at most $D$ from $S$. After the second step, no vertex colored $c_2$ of $C$ is at a distance smaller or equal than $D-1$, since otherwise such vertex cannot be adjacent to a vertex colored $c_1$. Finally, after the last step, a vertex remains colored $c_0$ if it was adjacent to a vertex colored $c_1$, and by argument above, it must be at a distance at least $D-1$ from $S$.

Hence, if we consider the graph $G'$ composed of all the vertices at distance at most $D-2$ from $S$, then, the coloring $\alpha'$ is a $(k-1)$-coloring of $G$, and there is a recoloring sequence from $\alpha'|_S$ to $\beta_S$ using Kempe chains of length at most $A-1$. Moreover, the distance from $S$ to the border of $G'$ is at least $D-3 \geq 3(A-1)$. Using the induction hypothesis, there is an extension $\beta$ of $\beta_S$ and a recoloring sequence from $\alpha'$ to $\beta$ which does not recolor the border of $G'$. This gives a recoloring sequence from $\alpha$ to $\beta$ which does not recolor the border of $G$.

This completes the induction step and proves the property. The lemma then follows immediately from the existence of a recoloring sequence using Kempe chains of length at most $n$ for any interval graph with $n$ vertices.
\end{proof}

Let $x$ be a vertex, $N$ be an integer and $V'$ be the set of vertices at distance at least $N+1$ from $x$. The $N$ vertices \emph{before} (resp. \emph{after}) $x$ in the ordering are the $N$ vertices before (resp. after) $x$ in a clique path of $V'$ where the vertices with neighbors in $V \setminus V'$ are in the first or the last bag of the clique path in the ordering.

The following is a corollary of Lemma~\ref{lem:mainkempe}:

\begin{restatable}{corollary}{CoroKempe}
\label{coro:kempe}
Let $G$ be an interval graph, and $S$ be a $(4,5)$-ruling set of $G$. Let $N$ be an integer. Let $\alpha,\beta$ be two colorings of $G$ and $c_0$ be a color that does not appear in $\alpha,\beta$.
Let $I$ be an independent set at distance $6(2N+1)+1$ of the virtual path induced by $S$~\footnote{Note that since consecutive vertices of the virtual path are at distance at most $11$, we can compute it in $O(\log^*(n))$ rounds.}
Then, it is possible to recolor, for every $v \in I$,  $v$ and the $N$ vertices before and after it
in the interval graph with their final coloring in $\beta$ with a schedule of length at most $3(2N+1)$.
\end{restatable}

\begin{proof}
Let $I$ be an independent set at distance $6(2N+1)+1$ of the virtual path induced by $S$ and $v \in I$. We denote by $S_v$ the set of vertices in the box of $v$ and the $N$ vertices before and after $v$ in the interval graph. Note that $S_v$ has size at most $(2N+1)$. Let us denote by $G_v$ the graph between $v'$ and $v''$, where $v'$ and $v''$ are the elements of $S$ that are at distance $3(2N+1)$ before and after $v$ in $S$. Note that by definition of $(4,5)$-ruling set, the box of $v'$ and the box of $v''$ are at distance at least $3(2N+1)$ from $v$. So by Lemma~\ref{lem:mainkempe}, it is possible to recolor $S_v$ with its target coloring in $\beta$ by only recoloring vertices of $G_v$ which are not in the border of $G_v$. Moreover, the length of the schedule is at most $3(2N+1)$.

One can notice that since $I$ is an independent set of $S$ at distance $6(2N+1)$, for every $v,w$ in $I$, $G_v$ and $G_w$ do not intersect. So we can perform in parallel all these recoloring sequences without any conflict, which completes the proof.
\end{proof}

\subsection{Buffer technique}\label{sec:buffer}

Let $G$ be an interval graph and $S$ be a $(4,5)$-ruling set of $G$.  Let $\alpha,\beta$ be two colorings of $G$ such that $\alpha$ is a $k$-coloring with $k < 4\omega$ and $\beta$ is an $(\omega+1)$-coloring. 
Let $N:=3 \Delta^2 (3 {\omega \choose 2} +2 + \omega)+2\omega$. 

Let $I$ be an independent set of the virtual path of $S$ at distance at least $6(2N+1)+1$ in the ruling path. By Corollary~\ref{coro:kempe}, using the additional color, we can recolor the graph in such a way that, for every $v \in I$, $v$ plus the $N$ vertices before and after them in the interval graph are colored with their target color in the $(\omega+1)$-coloring $\beta$. Let us call such a coloring $\gamma$ a \emph{semi-target coloring} for $\beta$. 
Informally speaking, a semi target coloring is a coloring such that, at bounded distance to every vertex of the graph (namely $O(N)$), there exist at least $N$ consecutive vertices that are colored with the  target coloring $\beta$. The rest of the proof consists in proving that we can recolor the graph between two consecutive such sets with the target coloring without modifying the rest of the coloring. In other words, if we have two sets of $N$ consecutive vertices that are colored with the target coloring, we can recolor the subgraph between these vertices with the target coloring without any knowledge on the rest of the graph.

After proving this result, we will be able to derive the main result of this section:

\begin{restatable}{lemma}{ClaimSemiGood}\label{clm:semigood}
Let $\gamma,\beta$ be two colorings of $G$ such that $\beta$ is a $(\omega+1)$-coloring and $\gamma$ is a semi-target coloring for $\beta$. Then we can recolor $\gamma$ into $\beta$ as long as $k \ge \omega+4$ (without additional colors) with a schedule of length at most $O(\poly(\Delta))$ in $O(N \logsn)$ rounds in the LOCAL model.
\end{restatable}

The rest of this section is devoted to prove Lemma~\ref{clm:semigood}. In order to prove it, we will need the notion of buffer coloring introduced in~\cite{BousquetB19}. In this paper, we use the result of~\cite{BousquetB19} as a black-box and do not precisely define the notion of buffer coloring (due to the fact that the definition is quite long and technical).

Let $N'=N-2\omega \cdot (\Delta)$ and $N''=N'/\omega-1$. 

\begin{restatable}{claim}{Goodlength}\label{clm:goodlength}
Let $u,v$ be two vertices of $I$ such that the interbox $H$ of $u,v$ contains $N$ consecutive vertices colored with the target coloring $\beta$. Then any clique path of $H$ contains $N''+1$ consecutive cliques colored with the target coloring.
\end{restatable}
\begin{proof}
Let $B$ be the set of (at least $N$) consecutive vertices that are colored with the target color. Let $A$ be the border of $B$. Let $X$ be the set of vertices of $B$ which are not in $A$ nor adjacent to a vertex of $A$. Note that $A$ has size at most $2\omega$ since the border of a set of consecutive vertices is  composed of at most two cliques. Then $X$ has size at least $N'=N-2\omega (\Delta+1)$. Since every bag of the clique path contains at most $\omega$ vertices, $A$ contains the vertices of at least $N''+1=N'/\omega$  consecutive cliques of a clique path of $B$. 
\end{proof}

Let $H$ be an interbox containing $N$ consecutive vertices colored with the target coloring in the interbox.
Let $\mathcal{P}$ be a path decomposition of $H$. By Claim~\ref{clm:goodlength}, there exists a subset $Y'$ of $N''+1$ cliques colored with the target coloring $\beta$. Let $Y$ be the last $N''$ cliques of them. Let us denote by $C$ and $D$ the first and last clique of $Y$. Note that in particular, the clique before $C$ is colored with the target coloring $\beta$. 

\paragraph*{Buffer colorings.}
A coloring $\delta$ restricted to $N''$ consecutive cliques $Y$ of a clique path $\mathcal{P}$ (given with an orientation) is a \emph{buffer coloring (for $Y$)} if it satisfies several properties detailed in~\cite{BousquetB19}. Let us denote by $C,D$ the first and last clique of $Y$. Due to the technicality of the definition, we do not state these properties precisely here and will only use the existence of such a coloring as a black-box.
In our proof, we will only need the following properties that are either in the definition that $Y$ is a buffer coloring in~\cite{BousquetB19} or simple consequences of that definition:

\begin{enumerate}
    \item If $Y$ is colored with the target coloring and the clique before $Y$ in $\mathcal{P}$ is also colored with the target coloring, then the coloring is a buffer coloring for $Y$.
    \item If $\delta$ is a buffer coloring for $Y$, then the clique before $C$ in $\mathcal{P}$ is colored with the target coloring as well as the first $\Delta$ cliques of $Y$\footnote{The exact condition of~\cite{BousquetB19} is a bit weaker since it only needs that the clique before in the clique path is colored with the target coloring. To avoid the technicality of clique path, we ask for this stronger condition here.}.
    \item  If $\delta$ is a buffer coloring for $Y$ and the vertices of $V\setminus X$ incident to $C \cup D$ are colored with the target coloring, then $\delta$ can be transformed into the target coloring on $X$ by recoloring only vertices of $X$.
\end{enumerate}

Using buffer colorings, Bousquet and Bartier~\cite{BousquetB19} proved the following:

\begin{restatable}{lemma}{LemBuffer}
\label{lem:buffer}
Let $G$ be an interval graph with clique path $\mathcal{P}$ (given with an ordering). Let $\delta,\beta$ be two colorings. Let $k'$ be the number of colors in $\beta$ and $k \ge k'+3$.

Let $Y$ be a set of $N''$ consecutive cliques of $G$ such that $\delta$ is a buffer coloring for $Y$ (for the target coloring $\beta$). Let $C$ be the first clique of $X$ \footnote{Note that we have $\delta_C=\beta_C$ by the above remarks.}.
Then we can obtain a coloring $\delta'$ from $\delta$ such that:
\begin{itemize}
    \item Only vertices that belong to bags in $X$ are recolored.
    \item No vertex of $C$ is recolored. In particular $\delta'_C=\beta_C$.
    \item The coloring $\delta'$ restricted to the $N''$ cliques $Z$ starting in the clique after $C$ (that is, $Y$ minus its first clique plus the first clique after it) is a buffer coloring for $Z$.
    \item The total length of the schedule is $\poly(|Y|,k)$.
\end{itemize}
\end{restatable}

\begin{proof}[Sketch of the proof.]
The proof is a direct consequence of the Lemmas 20, 21, 22 and~23 of~\cite{BousquetB19}. The core of the proof of Theorem 24 contains a proof of this statement, even though it is not explicitly stated. At first glance one can remark a few differences with the statement of~\cite{BousquetB19}. In~\cite{BousquetB19}, they need the whole coloring before $S$ to be the target coloring, but since we only recolor vertices of $S$, the condition that the neighbors of $X$ are colored with the target coloring is enough. Moreover, their number of colors is only $\omega+3$ there while we need $\omega+4$. The reason is due to the fact that we are targeting a $\omega+1$ instead of an optimal $\omega$-coloring and the number of colors of their proof is actually the target number of colors plus $3$. Note moreover that we have a reasoning on a number of vertices while it is on consecutive cliques in~\cite{BousquetB19} but, as we explained above the bound on the number of vertices here ensures that the result holds on the number of consecutive cliques of the clique path.
\end{proof}

Let us now explain how we can use buffer colorings to recolor interval graphs in a distributed way. To do so, we will prove the following as a consequence of Lemma~\ref{lem:buffer}:

\begin{lemma}\label{lem:recoluptoborder}
Let $G$ be an interval graph and $k \ge \omega+4$. Let $X$ be a subset of at least $N'$ consecutive vertices of $G$ and $B$ be the border of $X$ and $\delta,\beta$ be two colorings of $G$ such that:
\begin{itemize}
    \item $\beta$ is an $(\omega+1)$-coloring and,
    \item $\delta$ is a $k$-coloring such that $\delta$ admits $N'$ consecutive vertices $X$ which forms a buffer coloring for $\beta$,
    \item $\beta_B=\delta_B$.
\end{itemize}

Then there is a recoloring schedule from $\delta$ to $\beta$ of length $O(\poly(\Delta,|X|))$ using no additional color and  which never recolors the vertices in $B$.
\end{lemma}
\begin{proof}
The proof is a consequence of Lemma~\ref{lem:buffer}. Since $X$ forms a buffer, we can iteratively apply Lemma~\ref{lem:buffer} from $X$ going to the right until we reach the right border of $G$ (without recoloring it). 
We know that the colorings $\beta$ and $\delta$ agree on the right border. When we reach the right border, the coloring of the $N'$ vertices before it is a buffer coloring. And we observed that if the neighbors of a buffer are colored with the target coloring, then the coloring of the buffer can be modified to obtain the target coloring. So we can assume that $X$ and all the vertices at the right of $X$ are colored with the target coloring.

We can then repeat this argument from $X$ going to the left in order to recolor all the vertices at the left of $X$ with the target coloring, which completes the proof.
\end{proof}

We can now prove Lemma~\ref{clm:semigood} using Lemma~\ref{lem:recoluptoborder}.

\begin{proof}[Proof of Lemma~\ref{clm:semigood}]
Since $\delta$ is a semi-target coloring for $\beta$, for every $i \in I$, the $N$ vertices $X_i$ before and after every box of $i$ are colored with their target coloring. For every $i \in I$, Let $N'_i$ be the set of vertices of $X_i$ where vertices of the border of $X_i$ as well as their neighbors have been removed. Let $X_i'$ be the resulting set of size at least $N'$. The vertices of $X_i'$ are colored with the target coloring as well as their neighbors, so the semi-target coloring is a buffer coloring for the target coloring $\beta$.  

Now let $I'$ be  an independent set at distance $5$ in the ruling path of $I$. Such an independent set can be found in $O(\poly(\omega,\Delta) \logsn)$ rounds. We can number the elements of $I$ with label $0$ if they are in $I'$ and with label $i+1$ if they are incident to a label $i$. 

\begin{itemize}
    \item From every node $y$ of $I'$, we denote by $G_y$ the graph between the boxes of the closest labelled $2$ nodes of $y$. By Lemma~\ref{lem:recoluptoborder}, we can recolor  all the vertices between these two blocks with their target coloring. Moreover, since all the graphs $G_y$ are disjoint and not incident (there is a node labeled $3$ between any two consecutive labelled $0$ vertices), all these recolorings can be performed simultaneously. So a $O(\poly(N,k,\Delta))$ schedule is enough for this step.
    \item Now, consider the sections between two labeled $1$ vertices. Since all the vertices between the labeled $0$ and the labeled $2$ blocks are now colored with the target coloring, we indeed have a buffer coloring around labeled $1$ vertices. So by Lemma~\ref{lem:recoluptoborder} we can recolor all the vertices between these blocks with the target coloring. Again, all these operations can be performed simultaneously. Hence, a $O(\poly(N,k,\Delta))$ schedule is enough for this step.
\end{itemize}
After these two steps, (i) the nodes between a $0$-labelled and a $2$-labelled node of $I'$ are colored with the target coloring and (i) the nodes between two consecutive $1$-labelled nodes of $I'$ are colored with the target coloring. So the whole graph is colored with the target coloring, which completes the proof.
\end{proof}

Using Lemma~\ref{lem:recoluptoborder} we can prove the following lemma that will be used for chordal graphs.

\begin{lemma}\label{lem+interval+border}
Let $G$ be an interval graph and $\gamma,\beta$ be two colorings of $G$ such that:
\begin{enumerate}
    \item $\gamma$ is a $(\omega+3)$-coloring,
    \item $\beta$ is an $(\omega+1)$-coloring,
    \item $\gamma$ and $\beta$ agree on the leftmost and rightmost cliques of $G$,
    \item the diameter of $G$ is at least $12(2N + \Delta + 1) + 3$.
\end{enumerate}
Let $T$ be an upper bound on the diameter of $G$, and assume that we have access to one additional color. 
Then there is a recoloring schedule from $\gamma$ to $\beta$ of length $O(\poly(|G|))$ which never recolors the vertices in $B$ in at most $O(T)$ rounds.
\end{lemma}
\begin{proof}
By Lemma~\ref{lem:recoluptoborder}, if we can transform $\gamma$ into a coloring for which $N'$ consecutive vertices form a buffer coloring for the target coloring $\beta$ the conclusion follows.

Since the diameter of $G$ is at least  $12(2N + \Delta + 1) + 3$, we can find a vertex $y$ at distance at least  $6(2N + \Delta + 1) + 1$ from the border of $G$ in $O(T)$ rounds (since in $O(T)$ rounds we have the full knowledge of the graph).
Now using the same techniques as in Corollary~\ref{coro:kempe}, we can recolor the $N$ vertices before and after $y$ with the target coloring, which gives, as we already observed, a buffer coloring for $\beta$.
\end{proof}

\section{Decomposing, coloring and recoloring chordal graphs}
\label{sec:chordal}

We now investigate how to extend the results from the previous sections to chordal graphs. 
The algorithms for coloring and recoloring proceed in roughly two steps. The first step computes a partition of the chordal graph, obtained by iteratively removing interval subgraphs with controlled diameter. During the computation of this decomposition, we will assume that the nodes of the graph have access to a local view of a clique tree of the graph $G$. These local views can be computed in a distributed fashion with a constant number of rounds using the method from~\cite{KonradZ19} (see Section~\ref{sec:chordal-decomp} for more details). 

In the second step, the coloring algorithm consists in adding back step by step the vertices which have been deleted in the previous phase. Informally speaking, for each connected component, either the vertices that have been removed are attached to a single clique, and then we know that all the vertices attached to that clique form a $(\omega-1)$-degenerate graph, and we can then give them a color between $1$ and $\omega$ greedily. Either those vertices are attached to two cliques at large enough distance, and we can then use the machinery introduced in Section~\ref{sec:decompo-coloring-intervals} to color them with $\omega+1$ colors in total. Since at each step, each set has bounded diameter, one leader can decide of the coloring for all the vertices of that set.
For the recoloring algorithm, the idea is almost the same. The algorithm constructs iteratively a recoloring schedule by adding back step by step the vertices which were removed in the previous phase either using degeneracy or the tools of Section~\ref{sec:recoloring-intervals}.

\subsection{Interval decomposition of chordal graphs}\label{sec:chordal-decomp}

\begin{definition}
Let $G$ be a chordal graph, and $H$ a subgraph of $G$. We say that $H$ is a \emph{pending interval graph} if $H$ is an interval graph, and the border of $H$ is a subset of the first clique in a clique path of $H$. We say that $H$ is a \emph{separator interval graph}, if it is an interval graph, and the vertices of the border all belong to either the first or the last clique of a clique path of $H$.

An \emph{interval decomposition} of width $D$ and depth $\ell$ of a chordal graph $G$ is a partition $V_1, \ldots, V_\ell$ of the vertices of $G$ such that, $G[V_i]$ is a disjoint union of interval graphs with diameter at most $3D$, and every connected component of $G[V_i]$ is either:
\begin{itemize}
    \item a pending interval in $G[V_1,...,V_i]$; 
    \item or a separator interval graph with diameter at least $D$, in $G[V_1,...,V_i]$.
\end{itemize}

\end{definition}

The result below is adapted from the proofs of~\cite{KonradZ19}. We redo the proof for completeness.

\begin{restatable}{lemma}{LemDecompChordal}
\label{lem+decomp+chordal}
For every $D \geq 0$, there exists a distributed algorithm which computes an interval decomposition of width $D$ and depth $O(\log n)$ in $O(D \log n)$ rounds in the LOCAL model.
\end{restatable}

\begin{proof}
The algorithm starts by computing the clique tree of $G$. This can be done in a constant number of steps using the algorithm from~\cite{KonradZ19}. In particular, every node $v$ knows exactly which nodes of the clique tree it belongs to and which vertices of the graph are in each of these bags.
Then the algorithm proceeds in two steps. In the first step, it computes a partition of $G$, such that each connected component in a given class of the partition is either a pending interval graph, or a separating interval graph of diameter at least $D$. In a second step, the algorithm splits further each element of the partition into three parts, to obtain the $3D$ upper bound on the diameter of the components.

In order to compute this first partition, the algorithm iteratively removes vertices from the graph, and deletes the corresponding nodes from the clique tree. In other words, it removes a node of the clique tree when the corresponding bag is empty and contract edges of the clique tree for which a bag is included in another.
At the $i^{th}$ iteration, the algorithm starts by pre-selecting the vertices to remove. The algorithm selects all the vertices which appear only in vertices of degree at most $2$ in the clique tree and: 
\begin{itemize}
    \item are part of a pending interval graph of diameter at most $3D$ or,
    \item are part of an interval graph of diameter at least $D$.
\end{itemize}

Such an iteration takes $O(D)$ rounds. In other words, we remove a vertex $v$ if, when we consider the clique tree restricted to all the vertices at distance at most $D$ from $v$, the vertex $v$ only appears in a subpath of the corresponding clique tree and none of the vertices of this path are branching nodes of the clique tree.

At this point of the algorithm, the selected vertices induce a disjoint union of interval graphs, and each of these interval graphs is either a pending interval graph, or a separating interval graph of diameter at least $D$. This process terminates after at most $O(\log n)$ iterations. 
Indeed, from the point of view of the clique tree, the peeling process we use is a version of the rake-and-compress of Miller and Reif~\cite{MillerR89}, for which it is well-known that we have removed all the vertices from the graph after at most  $\ell = O(\log n)$ steps (recall that the clique tree has at most $n$ nodes). 
Let $V_i$ be the vertices removed at the iteration $\ell - i$ of the algorithm. Since each iteration can be done in $O(D)$ steps, this part of the algorithm stops after $O(D \log n)$ rounds.

In the second part, we split again (in parallel) each $V_i$ into $W^i_1, W^i_2, W^i_3$, to ensure that each component has diameter $O(D)$. 
Note that these graphs are interval graphs, therefore we can use the machinery of Section~\ref{sec:decompo-coloring-intervals}.
First, all the components of $V_i$ with diameter at most $3D$ are left unmodified, and are added to $W^i_3$. 
For the components of diameter greater than $3D$, we can compute a box decomposition in time $O(\logsn)$ (Lemma~\ref{lem:compute-box-decomposition}). Then, as before, we can easily refine this decomposition in order to keep only boxes at distance between $D$ and $2D$ one from the other. The nodes of the boxes (of the refined decomposition) are added to $W^i_2$. The interboxes (of the refined decomposition) have diameter at most $2D$. The component which contained the border of the large interval graph are added to $W^i_1$, and all the remaining vertices are added to $W^i_3$. Since we can do these operations in parallel for each component, this second step finishes after $O(D \log^* n)$ steps. After these operations, we obtain an interval partition of $G$ of width $D$ and depth $O(\log n)$. \end{proof}

\subsection{$(\omega+1)$-coloring chordal graphs}

In this section, we show how an interval decomposition can be used to compute an $(\omega+1)$-coloring of chordal graphs in the LOCAL model. Namely, we prove the following theorem.

\ThmColChordal*

The proof of Theorem~\ref{thm:col-chordal} follows from the decomposition of Lemma~\ref{lem+decomp+chordal} and the lemma below.

\begin{lemma}
There is an algorithm that, given a graph $G$ and an interval decomposition of $G$ of depth $\ell$ and width at least $15(\omega+1)$, 
computes a $(\omega+1)$ coloring of $G$ in $O(\omega \ell)$ rounds in the LOCAL model.
\end{lemma}

\begin{proof}
Let $V_1, \ldots, V_{\ell}$ be the interval partition of width $\omega$, and let $G_i = G[V_1\cup \ldots \cup V_{i}]$. We start with a coloring $\sigma_1$ of $G_1$ which can be computed in $O(\omega)$ rounds since each connected component of $G_1$ has diameter $O(\omega)$. Then for each $2 \leq i \leq \ell$, the algorithm computes an extension of the coloring from $G_{i-1}$ to $G_i$, using again $O(\omega)$ rounds. At each step $i$, let us denote by $\sigma_{i-1}$ the current coloring of the graph, and consider a connected component $S$ of $G[V_i]$. Let $G' = G_i[S \cup N(S)]$. By construction, $G'$ is an interval graph, with some vertices in its border already colored. There are two possible cases:
\begin{itemize}
    \item $S$ is a pending interval graph, which means that the vertices of $G'$ already colored are on a single side and form a clique. Then, given any $\omega$-coloring $\eta$ of $G'$, we can assume, up to permuting the colors in $\eta$, that $\eta$ and $\sigma_{i-1}$ agree on the already colored vertices of $G'$ and use the coloring $\eta$ to extend the coloring to $G'$. Namely, the vertices of the pending interval graph chooses a leader which knows the whole pending interval graph and can then try all the possible extensions of $\sigma_{i-1}$ to the vertices of the pending interval graph using at most $\omega$ colors. Such an extension must exist and then the leader can just assign it to the vertices of the pending interval.
    \item $S$ is a separating interval graph with diameter at least $15(\omega+1)$. It means that we have at least $3(\omega+1)$ boxes in the interval. Then $G'$ is an interval graph, and the vertices already colored are on its border. By Lemma~\ref{lem:coloring-interval-border}, we know that there exists an $(\omega+1)$-coloring of $G'$ which extends the partial coloring given by $\sigma_{i-1}$.
\end{itemize}
In all cases, we can extend the coloring $\sigma_{i-1}$ into a coloring $\sigma_i$ where all the vertices in $V_i$ are colored. Since we use $O(\omega)$ rounds for each step $i$, the algorithm uses in total $O(\omega \ell)$ rounds in the LOCAL model.
\end{proof}

\subsection{Recoloring chordal graphs}

In this section we describe how to use an interval decomposition in order to compute a recoloring schedule. We prove the following theorem:

\ThmChordal*

\medskip

Theorem~\ref{thm:recol-chordal} uses our decomposition (Lemma~\ref{lem+decomp+chordal}) and the result below.

\begin{restatable}{lemma}{LemIntervalDecomp}
Let $G$ be an interval graph and $\alpha,\beta$ be two proper $\omega+3$-colorings of $G$. We suppose we are given an interval decomposition of width $D= \Omega(\omega^2 \Delta^2)$ and depth $\ell$ of $G$. It is possible to find a schedule of length at most $O(\poly(\Delta)^\ell)$ to transform $\alpha$ into $\beta$ in $O(D \ell)$ rounds in the LOCAL model.
\end{restatable}

\begin{proof}
Let $D$ chosen polynomial in $\omega$ and $\Delta$ (its value will be  specified later). Note that an interval graph of diameter $D=\poly(\omega, \Delta)$ has a number of nodes that is $O(\poly \Delta)$, because this number is upper bounded by $D\Delta$ and $\omega \leq \Delta$.  
Also let $k'$ be the total of number of colors used (including the additional colors).
Let us denote by $V_1, \ldots, V_\ell$ the interval decomposition of $G$, and $G_i = G[V_1\cup \ldots V_i]$. Let $\alpha$ and $\beta$ the two colorings given as input, and let $\alpha_i$ to $\beta_i$ be the restrictions of $\alpha$ and $\beta$ to $G_i$. The algorithm proceeds by building successively a recoloring schedule $\lambda_i$ for $G_i$ from $\alpha_i$ to $\beta_i$, and progressively extending this recoloring schedule for the rest of the graph. 
In order to simplify the proof, we will require that the schedule produced by the algorithm has an additional property, namely that all the vertices recolored at a step $j$ of the schedule, are recolored with the color $j \mod k'$. Note that we can easily produce a schedule satisfying this property, up to multiplying its length by a factor of $k'$. This property will be useful later in order to extend progressively the schedule.

Since $G_1$ is a disjoint union of interval graphs of diameter $O(D)$, we can compute a recoloring schedule in $G_1$ from $\alpha_1$ to $\beta_1$ of length linear in the number of nodes by~\cite{BousquetB19}, which is in $O(\poly \Delta)$ in $O(D)$ rounds.

Assume that we have computed a recoloring schedule in $G_{i-1}$. We now describe how to extend this schedule to $G_i$. Before each step of $\lambda_{i-1}$, we insert $t$ steps during which only the vertices of $V_i$ are recolored, where $t = \poly(\Delta)$ is the length of the schedule produced by Lemma~\ref{lem+interval+border} 
(multiplied by $k'$ to ensure the fact that the recolored vertices in each step go to the same color). Finally, after the last step of $\lambda_{i-1}$, we insert again $t$ steps where only the vertices of $V_i$ are recolored such that they can reach their target coloring.

Let us consider the step $j$ of $\lambda_{i-1}$,  and let $\sigma$ be the coloring of $G_{i}$ just before this step. Let $c = j \mod k'$ be the color with which vertices recolored at this step are recolored with. We know that $\sigma|_{V_i}$ uses only $\omega+3$ colors. Our goal is to recolor the vertices of $V_i$ so that no vertex is colored $c$, and the recoloring step of $\lambda_{i-1}$ can be safely applied. 

Since $G_{i-1}$ is obtained from $G_{i}$ by removing pending interval graphs, and separating interval graphs of diameter at least $D$, there is a coloring $\sigma'$ of $G_i$ which agrees with $\sigma$ on $G_{i-1}$ and which uses at most $\omega + 1$ colors, all different from $c$ for the vertices in $V_i$. 
Taking $D = 240 \omega^2 \Delta^2$, the conditions of Lemma~\ref{lem+interval+border} are satisfied, and there is a recoloring schedule from $\sigma$ to $\sigma'$ of length $t$. 
Similarly, after the last step of $\lambda_{i-1}$, if $\sigma$ is the current coloring, then only $\omega +1$ colors are used in $V_i$, and by Lemma~\ref{lem+interval+border}, there is a recoloring schedule from $\sigma$ to $\beta_i$ in $G_i$ of length $t$.

This new schedule can be computed for each component of $G[V_i]$ in a centralized way, and up to multiplying the length of the schedule by $k'$, we can assume that at step $j$, vertices are recolored with the color $j \mod k'$. In a distributed setting, this requires $O(D)$ steps since each component has diameter $O(D)$. Since we need $O(D)$ steps to extend the recoloring schedule from $G_{i-1}$ to $G_i$, the algorithm uses at most $O(D \ell)$ rounds in total. Moreover, the schedule has length at most $t^{\ell} = (\poly \Delta)^\ell$.
\end{proof}

Now, for Theorem~\ref{thm:recol-chordal}, from Lemma~\ref{lem+decomp+chordal}, we get $\ell$ in $O(\log n)$, and we set $D$ in $O(\omega^2\Delta^2)$, hence we have an algorithm in $O(\omega^2\Delta^2 \log n)$ rounds that produces a schedule of length $(\poly \Delta)^{\log n}$ that is $n^{O(\log \Delta)}$.

\section*{Acknowledgments}

We thank the reviewers of an earlier version for their comments, and Marthe Bonamy for starting this project with us.
The second author thanks Fabian Kuhn and Václav Rozhoň for their kind and expert answers to his questions on network decomposition. 

\bibliography{distrecol}

\end{document}